\documentclass[runningheads]{llncs}

\usepackage{listings}
\usepackage{cite}
\usepackage{amsmath,amssymb,amsfonts}
\usepackage{algorithm}
\usepackage{algpseudocode}%
\usepackage{textcomp}
\usepackage{pgfplots}
\usepgfplotslibrary{fillbetween}
\usetikzlibrary{patterns}
\usetikzlibrary{plotmarks}
\usepackage{url}
\usepgfplotslibrary{units} 
\usepackage{tikz}
\usepackage{subfig}
\usepackage{xcolor, colortbl}
\usepackage[subtle]{savetrees}
\usepackage[T1]{fontenc}
\usepackage{lmodern}

\usepackage{graphicx}
\usepackage{hyperref}
\usepackage[T1]{fontenc}
\usepackage[strings]{underscore}
\usepackage[subtle]{savetrees}

\newcommand{\ignore}[1]{}
\newcommand{\revision}[1]{#1}
\newcommand{\guessing}[0]{\#\mathsf{guessing}(pw)}
\newcommand{\hard}[1]{\ensuremath{\mathsf{GetHardness}(#1)}}

\newenvironment{remindertheorem}[1]{\medskip \noindent {\bf Reminder of  #1.  }\em}{}

\newenvironment{proofof}[1]{\begin{trivlist} \item {\bf Proof
#1:~~}}
  {\qed\end{trivlist}}
\renewenvironment{proofof}[1]{\par\medskip\noindent{\bf Proof of #1: \ }}{\hfill$\Box$\par\medskip}

\algrenewcommand\algorithmicrequire{\textbf{Input:}}
\algrenewcommand\algorithmicensure{\textbf{Output:}}

\newcommand{\red}[1]{\cellcolor{red!50}#1}
\newcommand{\yellow}[1]{\cellcolor{yellow!50}#1}

\newcommand{\tikzDefaults}[0]{
  \pgfplotsset{
    every axis/.append style={line width = 0.3pt},
    title style={align=center},
    xlabel={$v/C_{max}$},
    ylabel={Fraction of Cracked Passwords},
    log basis x={10},
    grid=major,
    cycle list = {{black}, {red, densely dotted}, {blue, densely dashed},{red!50!black, loosely dotted}, {blue!50!black, loosely dashdotted}, {blue, densely dashdotted}, {blue, densely dashdotdotted}, {blue!50!black, loosely dotted}},
    legend style = {font=\tiny, at={(.01,.99)}, anchor=north west},
    legend entries = {deterministic\\  $\tau = 3$\\ $\tau = 5$ \\ improvement: {{\color{black}black- \color{red} red}} \\ improvement: {{\color{black}black- \color{blue} blue}}\\}
    }
}

\begin{document}
\title{DAHash: Distribution Aware Tuning of Password Hashing Costs}
\titlerunning{DAHash}
\author{Wenjie Bai\inst{1} \and
Jeremiah Blocki\inst{1}}
\authorrunning{Wenjie Bai, Jeremiah Blocki}
%
\institute{Department of Compouter Science, Purdue University, IN, USA \\
\email{\{bai104,jblocki\}@purdue.edu}}
\maketitle              
\begin{abstract}
An attacker who breaks into an authentication server and steals all of the cryptographic password hashes is able to mount an offline-brute force attack against each user’s password. Offline brute-force attacks against passwords are increasingly commonplace and the danger is amplified by the well documented human tendency to select low-entropy password and/or reuse these passwords across multiple accounts. Moderately hard password hashing functions are often deployed to help protect passwords against offline attacks by increasing the attacker's guessing cost. However, there is a limit to how ``hard'' one can make the password hash function as authentication servers are resource constrained and must avoid introducing substantial authentication delay. Observing that there is a wide gap in the strength of passwords selected by different users we introduce DAHash (Distribution Aware Password Hashing) a novel mechanism which reduces the number of passwords that an attacker will crack. Our key insight is that a resource-constrained authentication server can dynamically tune the hardness parameters of a password hash function based on the (estimated) strength of the user's password. We introduce a Stackelberg game to model the interaction between a defender (authentication server) and an offline attacker. Our model allows the defender to optimize the parameters of DAHash e.g., specify how much effort is spent to hash weak/moderate/high strength passwords. We use several large scale password frequency datasets to empirically evaluate the effectiveness of our differentiated cost password hashing mechanism. We find that the defender who uses our mechanism can reduce the fraction of passwords that would be cracked by a rational offline attacker by \revision{around $15\%$}.


\vspace{-0.3cm}
\end{abstract}
\section{Introduction}

Breaches at major organizations have exposed billions of user passwords to the dangerous threat of offline password cracking. An attacker who has stolen the cryptographic hash of a user’s password could run an offline attack by comparing the stolen hash value with the cryptographic hashes of every password in a large dictionary of popular password guesses. An offline attacker can check as many guesses as s/he wants since each guess can be verified without interacting with the authentication server. The attacker is limited only by the cost of checking each password guess i.e., the cost of evaluating the password hash function. 

Offline attacks are a grave threat to security of users’ information for several reasons. First, the entropy of a typical user chosen password is relatively low~e.g., see \cite{SP:Bonneau12}. Second, users often reuse passwords across multiple accounts to reduce cognitive burden. Finally, the arrival of GPUs, FPGAs and ASICs significantly reduces the cost of evaluating a password hash functions such as PBKDF2~\cite{kaliski2000pkcs} millions or billions of times. Blocki et al.~\cite{SP:BloHarZho18} recently argued that PBKDF2 {\em cannot} adequately protect user passwords without introducing an intolerable authentication delay (e.g., $2$ minutes) because the attacker could use ASICs to reduce guessing costs by many orders of magnitude.  

Memory hard functions (MHFs)~\cite{Per09,Argon2} can be used to build ASIC resistant password hashing algorithms. The Area x Time complexity of an ideal MHF will scale with $t^2$, where $t$ denotes the time to evaluate the function on a standard CPU. Intuitively, to evaluate an MHF the attacker must dedicate $t$ blocks of memory for $t$ time steps, which ensures that the cost of computing the function is equitable across different computer architectures i.e., RAM on an ASIC is still expensive. Because the ``full cost'' \cite{JC:Wiener04} of computing an ideal MHF scales quadratically with $t$ it is also possible to rapidly increase guessing costs without introducing an untenable delay during user authentication --- by contrast the full cost of hash iteration based KDFs such as  PBKDF2~\cite{kaliski2000pkcs}  and BCRYPT~\cite{provos1999bcrypt} scale linearly with $t$. Almost all of the entrants to the recent Password Hashing Competition (PHC)~\cite{EPRINT:Wetzels16} claimed some form of memory-hardness.  

Even if we use MHFs there remains a fundamental trade-off in the design of good password hashing algorithms. On the one hand the password hash function should be sufficiently expensive to compute so that it becomes economically infeasible for the attacker to evaluate the function millions or billions of times per user --- even if the attacker develops customized hardware (ASICs) to evaluate the function.  On the other hand the password hashing algorithm cannot be so expensive to compute that the authentication server is unable to handle the workload when multiple users login simultaneously. Thus, even if an organization uses memory hard functions it will not be possible to protect all user passwords against an offline attacker e.g., if the password hashing algorithm is not so expensive that the authentication server is overloaded then it will almost certainly be worthwhile for an offline attacker to check the top thousand passwords in a cracking dictionary against each user's password. In this sense all of the effort an authentication server expends protecting the weakest passwords is (almost certainly) wasted. 
\vspace{-0.3cm}

\paragraph{Contributions} 
We introduce DAHash (Distribution Aware Hash) a password hashing mechanism that minimizes the damage of an offline attack by tuning key-stretching parameters for each user account based on password strength. In many empirical password distributions there are often  several passwords that are so popular that it would be infeasible for a resource constrained authentication server to  dissuade an offline attacker from guessing these passwords e.g., in the Yahoo! password frequency corpus ~\cite{SP:Bonneau12,NDSS:BloDatBon16} the most popular password was selected by approximately $1\%$ of users. Similarly, other users might select passwords that are strong enough to resist offline attacks even with minimal key stretching. The basic idea behind DAHash is to have the resource-constrained authentication server shift more of its key-stretching effort towards saveable  password i.e., passwords the offline attacker could be disuaded from checking. 

  Our DAHash mechanism partitions passwords into $\tau$ groups e.g., weak, medium and strong when $\tau=3$. We then select a different cost parameter $k_i$ for each group $G_i$, $i \leq \tau$ of passwords. If the input password $pw$ is in group $G_i$ then we will run our moderately hard key-derivation function with cost parameter $k_i$ to obtain the final hash value $h$.  Crucially, the hash value $h$ stored on the server will not reveal any information about the cost parameter $k_i$ or, by extension, the group $G_i$.  

  We adapt a Stackelberg Game model of Blocki and Datta~\cite{BlockiD16} to help the defender (authentication server) tune the DAHash cost parameters $k_i$ to minimize the fraction of cracked passwords. The Stackelberg Game models the interaction between the defender (authentication server) and an offline attacker as a Stackelberg Game. The defender (leader) groups passwords into different strength levels and selects the cost parameter $k_i$ for each group of passwords (subject to maximum workload constraints for the authentication server) and then the offline attacker selects the attack strategy which maximizes his/her utility (expected reward minus expected guessing costs). The attacker's expected utility will depend on the DAHash cost paremeters $k_i$ as well, the user password distribution, the value $v$ of a cracked password to the attacker and the attacker's strategy i.e., an ordered list of passwords to check before giving up. We prove that an attacker will maximize its utility by following a simple greedy strategy. We then use an \revision{evolutionary} algorithm to help the defender compute an optimal strategy i.e., the optimal way to tune DAHash cost parameters for different groups of passwords. The goal of the defender is to minimize the percentage of passwords that an offline attacker cracks when playing the utility optimizing strategy in response to the selected DAHash parameters $k_1,\ldots,k_{\tau}$. 
  
 Finally, we use several large password datasets to evaluate the effectiveness of our differentiated cost password hashing mechanism. \revision{We use the empirical password distribution to evaluate the performance of DAHash when the value $v$ of a cracked password is small. We utilize Good-Turing frequency estimation to help identify and highlight uncertain regions of the curve i.e., where the empirical password distribution might diverge from the real password distribution. To evaluate the performance of DAHash when $v$ is large we derive a password distirbution from guessing curves obtained using the Password Guessing Service~\cite{USENIX:USBCCKKMMS15}. The Password Guessing Service uses sophisticated models such as Probabilistic Context Free Grammars~\cite{SP:WAMG09,SP:KKMSVB12,NDSS:VerColTho14}, Markov Chain Models~\cite{NDSS:CasDurPer12,Castelluccia2013,SP:MYLL14,USENIX:USBCCKKMMS15}  and even neural networks~\cite{USENIX:MUSKBCC16} to generate password guesses using Monte Carlo strength estimation~\cite{CCS:DelFil15}.  We find that DAHash reduces the fraction of passwords cracked by a rational offline attacker by up to $15\%$ (resp. $20\%$) under the empirical distribution (resp. derived distribution).  }
 \vspace{-0.3cm}

\section{Related Work}

\ignore{\subsection{Key-stretching} }
Key-stretching was proposed as early as 1979 by Morris and Thomson as a way to protect passwords against brute force attacks~\cite{Morris1979}. Traditionally key stretching has been performed using hash iteration e.g., PBKDF2~\cite{kaliski2000pkcs} and BCRYPT~\cite{provos1999bcrypt}. More modern hash functions such as SCRYPT and Argon2~\cite{Argon2}, winner of the password hashing competition in 2015~\cite{EPRINT:Wetzels16}, additionally require a significant amount of memory to evaluate. An economic analysis Blocki et al.~\cite{SP:BloHarZho18} suggested that hash iteration based key-derivation functions no longer provide adequate protection for lower entropy user passwords due to the existence of ASICs. On a positive note they found that the use of memory hard functions can significantly reduce the fraction of passwords that a rational adversary would crack. 

The addition of ``salt'' is a crucial defense against rainbow table attacks~\cite{C:Oechslin03} i.e., instead of storing $(u, H(pw_u))$ and authentication server will store $(u,s_u,H(s_u,pw_u))$ where $s_u$ is a random string called the salt value. Salting defends against pre-computation attacks (e.g.,~\cite{EC:DodGuoKat17}) and ensures that each password hash will need to be cracked independently e.g., even if two users $u$ and $u'$  select the same password we will have $H(s_{u'},pw_{u'}) \neq H(s_u,pw_u)$ with high probability as long as $s_u \neq s_{u'}$.  

Manber proposed the additional inclusion of a short random string called ``pepper'' which would not be stored on the server~\cite{manber1996simple} e.g., instead of storing  $(u,s_u,H(s_u,pw_u))$ the authentication server would store $(u,s_u,H(s_u,x_u,pw_u))$ where the pepper $x_u$ is a short random string that, unlike the salt value $s_u$, is not recorded. When the user authenticates with password guess $pw'$ the server would evaluate $H(s_u,x,pw')$ for each possible value of $x \leq x_{max}$ and accept if and only if $H(s_u,x,pw') = H(s_u,x_u,pw_u)$ for some value of $x$. The potential advantage of this approach is that the authentication  server can usually halt early when the legitimate user authenticates, while the attacker will have to check every different value of $x\in [1,x_{max}]$ before rejecting an incorrect password. Thus, on average the attacker will need to do more work than the honest server. 

Blocki and Datta observed that non-uniform distributions over the secret pepper value $x \in [1,x_{max}]$ can sometime further increase the attacker's workload relative to an honest authentication server~\cite{BlockiD16}. They showed how to optimally tune the pepper distribution by using  Stackelberg game theory~\cite{BlockiD16}. However, it is not clear how pepper could be effectively integrated with a modern memory hard function such as Argon2 or SCRYPT. One of the reasons that MHFs are incredibly effective is that the ``full cost''~\cite{JC:Wiener04} of evaluation can scale quadratically with the running time $t$. Suppose we have a hard limit on the running time $t_{max}$ of the authentication  procedure e.g., $1$ second. If we select a secret pepper value $x \in [1,x_{max}]$ then we would need to ensure that $H(s_u,x,pw')$ can be evaluated in time at most $t_{max}/x_{max}$ --- otherwise the total running time to check all of the different pepper values sequentially would exceed $t_{max}$. In this case the ``full cost'' to compute $H(s_u,x,pw')$ for every $x \in [1,x_{max}]$ would be at most $ O \left( x_{max} \times(t_{max}/x_{max})^2 \right) =  O \left( t_{max}^2/x_{max} \right) $.  If instead we had not used pepper then it would have been possible to ensure that the full cost could be as large as $\Omega(t_{max}^2)$ simply by allowing the MHF to run for time $t_{max}$ on a single input. Thus, in most scenarios it would  be preferable for the authentication server to use a memory-hard password hashing algorithm without incorporating pepper. 

Boyen's work on ``Halting Puzzles'' is also closely related to our own work \cite{USENIX:Boyen07}. In a halting puzzle the (secret) running time parameter $t \leq t_{max}$ is randomly chosen whenever a new account is created. The key idea is that an attacker will need to run in time $t_{max}$ to definitively reject an incorrect password while it only takes time $t$ to accept a correct password. In Boyen's work the distribution  over running time parameter $t$ was the same for all passwords. By contrast, in our work we assign a fixed hash cost parameter to each password and this cost parameter may be different for distinct passwords. We remark that it may be possible to combine both ideas i.e., assign a different maximum running time parameter $t_{max,pw}$ to different passwords. We leave it to future work to explore whether or not the composition of both mechanisms might yield further security gains. 

\vspace{-0.3cm}


\section{DAHash}

In this section, we first introduce some preliminaries about passwords then present the DAHash and explain how the authentication process works with this mechanism. We also discuss ways in which a (rational) offline attacker might attempt to crack passwords protected with the differentiated cost mechanism.

\vspace{-0.3cm}
\subsection{Password Notation}
We let $\mathcal{P} = \{pw_1,pw_2,\ldots,\}$ be the set of all possible user-chosen passwords. We will assume that passwords are sorted so that $pw_i$ represents the $i$'th most popular password. Let $\Pr[pw_i]$ denote the probability that a random user selects password $pw_i$ we have a distribution over $\mathcal{P}$ with $\Pr[pw_1] \geq \Pr[pw_2] \geq \ldots $ and $\sum_i \Pr[pw_i] = 1$. 

\revision{The distributions we consider in our empirical analysis have a compressed representation. In particular, we can partition the set of passwords  $\mathcal{P}$ into $n'$ equivalence sets $es_1, \ldots, es_{n'}$ such that for any $i$, $pw,pw' \in es_i$ we have $\Pr[pw] = \Pr[pw'] = p_i$. In all of the distributions we consider we will have $n' \ll \left| \mathcal{P}\right|$ allowing us to efficiently encode the distribution using $n'$ tuple $(|es_1|, p_1),\ldots, (|es_{n'}|, p_{n'})$ where  $p_i$ is the probability of any password in equivalence set $es_i$. We will also want to ensure that we can optimize our DAHash parameters in time proportional to $n'$ instead of $\left| \mathcal{P}\right|$. }

\vspace{-0.3cm}

\subsection{DAHash}\label{AA} 

{\bf Account Creation: } When a new user first register an account with user name $u$ and password $pw_u\in\mathcal{P}$ DAHash will first assign a hash cost parameter $k_u = \mathsf{GetHardness}(pw_u)$ based on the (estimated) strength of the user's password. We will then randomly generate a $L$ bit string $s_u\leftarrow\{0,1\}^L$ (a “salt”) then compute hash value $h_u=H\left(pw_u,s_u; k_u\right)$, at last store the tuple $\left(u,s_u, h_u\right)$ as the record for user $u$. The salt value $s_u$ is used to thwart rainbow attacks~\cite{C:Oechslin03} and $k_u$ controls the cost of hash function\footnote{We remark that the hardness parameter $k$ is similar to ``pepper'' \cite{manber1996simple} in that it is not stored on the server. However, the hardness parameter $k$ is distinct from pepper in that it is derived deterministically from the input password $pwd_u$. Thus, unlike pepper, the authentication server will not need to check the password for every possible value of $k$. }.  

 {\bf Authentication with DAHash:} Later, when user $u$ enters her/his password $pw_u'$, the server first retrieves the corresponding salt value $s_u$ along with the hash value $h_u$, runs $\mathsf{GetHardness}(pw_u')$ to obtain $k_u'$ and then checks whether the hash $h_u' = H(pw_u',s_u; ~k_u')$ equals the stored record $h_u$ before granting access. If $pw_u'= pw_u$ is the correct password then we will have $k_u'=k_u$ and $h_u'=h_u$ so authentication will be successful. Due to the collision resistance of cryptographic hash functions, a login request from someone claiming to be user $u$ with password $pw’_u\neq pw_u$ will be rejected. The account creation and authentication processes are formally presented in Algorithms \ref{alg:createaccount} and \ref{alg:authenticate} (see Appendix \ref{app:algorithm}).

 In the traditional (distribution oblivious) key-stretching mechanism $\mathsf{GetHardness}(pw_u)$ is a constant function which always returns the same cost parameter $k$. Our objective will be to optimize $\mathsf{GetHardness}(pw_u)$ to minimize the percentage of passwords cracked by an offline attacker. This must be done subject to any workload constraints of the authentication server and (optionally) minimum protection constraint, guiding the minimum acceptable key-stretching parameters for any password.

   The function $\mathsf{GetHardness}(pw_u)$ maps each password to a hardness parameter $k_u$ which controls the cost of evaluating our password hash function $H$.  For hash iteration based key-derivation functions such as PBKDF2 we would achieve cost $k_u$ by iterating the underling hash function $t = \Omega(k)$ times. By contrast, for an ideal memory hard function the full evaluation cost scales quadratically with the running time $t_u$ so we have $t_u = O\left( \sqrt{k_u}\right)$ i.e., the attacker will need to allocate $t_u$ blocks of memory for $t_u$ time steps. In practice, most memory hard functions will take the parameter $t$ as input directly. For simplicity, we will assume that the cost parameter $k$ is given directly and that the running time $t$ (and memory usage) is derived from $k$.

{\bf Remark.} We stress that the hardness parameter $k$ returned by $\mathsf{GetHardness}(pw_u)$ should not be stored on the server. Otherwise, an offline attacker can immediately reject an incorrect password guess $pw' \neq pw_u$ as soon as he/she observes that $k\neq \mathsf{GetHardness}(pw^\prime)$. Furthermore, it should not possible to directly infer $k_u$ from the hash value $h_u\leftarrow H(pw_u,s_u;~k_u)$. Any MHF candidate such as SCRYPT~\cite{Per09}, Argon2~\cite{Argon2} or DRSample~\cite{CCS:AlwBloHar17} will satisfy this property.  While the hardness parameter $k_u$ is not stored on the server, we do assume that an offline attacker who has breached the authentication server will have access to the function $\mathsf{GetHardness}(pw_u)$ (Kerckhoff's Principle) since the code for this function would be stored on the authentication server. Thus, given a password guess $pw'$ the attacker can easily generate the hardness parameter $k' = \mathsf{GetHardness}(pw')$ for any particular password guess. 

{\bf Defending against Side-Channel Attacks.}
   A side-channel attacker might try to infer the hardness parameter $k$ (which may in turn be correlated with the strength of the user's password) by measuring delay during a successful login attempt. We remark that for modern memory hard password hashing algorithms ~\cite{Per09,Argon2,CCS:AlwBloHar17} the cost parameter $k$ is modeled as the product of two parameters: memory and running time. Thus, it is often possible to increase (decrease) the cost parameter without affecting the running time simply by tuning the memory parameter\footnote{By contrast, the cost parameter for PBKDF2 and BCRYPT is directly proportional to the running time. Thus, if we wanted to set a high cost parameter $k$ for some groups of passwords we might have to set an intolerably long authentication delay~\cite{SP:BloHarZho18}.}. Thus, if such side-channel attacks are a concern the authentication server could fix the response time during authentication to some suitable constant and tune the memory parameter accordingly.  Additionally we might delay the authentication response for a fixed ammount of time (e.g., 250 milliseconds) to ensure that there is no correlation between response time and the user's password.
   
\vspace{-0.3cm}
\subsection{Rational Adversary Model}
We consider an untargeted offline adversary whose goal is to break as many passwords as possible. In the traditional authentication setting
an offline attacker who has breached the authentication server has access to all the data stored on the server, including each user's record $(u,s_u,h)$ and the code for hash function $H$ and for the function $\mathsf{GetHardness}()$. In our analysis we assume that $H$ can only be used as a black box manner (e.g., random oracle)  to return results of queries from the adversary and that attempts to find a collision or directly invert $H(\cdot)$ succeed with negligible probability. However, an offline attacker who obtains $(u,s_u,h)$ may still check whether or not $pw_u=pw'$ by setting $k' = \mathsf{GetHardness}(pw')$ and checking whether or not $h= H(pw',s_u;~k')$. The only limitation to adversary's success rate is the resource she/he would like to put in cracking users’ password. 

We assume that the (untargetted) offline attacker has a value $v=v_u$  for password of user $u$. For simplicity we will henceforth use $v$ for password value since the attacker is untargetted and has the same value $v_u=v$ for every user $u$. There are a number of empirical studies of the black market~\cite{CCS:Allodi17,goldForSilver,stockley2016} which show that cracked passwords can have substantial value e.g., Symantec reports that passwords generally sell for $\$4-\$30$~\cite{passwordBlackMarket} and \cite{stockley2016} reports that e-mail passwords typically sell for $\$1$ on the Dark Web. \revision{Bitcoin ``brain wallets'' provide another application where cracked passwords can have substantial value to attackers \cite{FC:VBCKM16}.}

We also assume that the untargetted attacker has a dictionary list which s/he will use as guesses of $pw_u$) e.g., the attacker knows $pw_i$ and $\Pr[pw_i]$ for each password $i$. However, the the attacker will not know the particular password $pw_u$ selected by each user $u$. Therefore, in cracking a certain user’s account the attacker has to enumerate all the candidate passwords and check if the guess is correct until there is a guess hit or the attacker finally gives up. We assume that the attacker is rational and would choose a  strategy that would maximize his/her expected utility. The attacker will need to repeat this process independently for each user $u$. In our analysis we will focus on an individual user's account that the attacker is trying to crack. 

\vspace{-0.3cm}

\section{Stackelberg Game}\label{sec: Game}
In this section, we use Stackelberg Game Theory~\cite{von2010market} to model the interaction between the authentication server and  an untargeted adversary so that we can optimize the DAHash cost parameters. In a Stackelberg Game the leader (defender) moves first and then the follower (attacker) plays his/her best response. In our context, the authentication server (leader) move is to specify the function $\mathsf{GetHardness}()$. After a breach the offline attacker (follower) can examine the code for $\mathsf{GetHardness}()$ and observe the hardness parameters that will be selected for each different password in $\mathcal{P}$. A rational offline attacker may use this knowledge to optimize his/her offline attack. We first formally define the action space of the defender (leader) and attacker (follower) and then we formally define the utility functions for both players. 
\vspace{-0.3cm}
\subsection{Action Space of Defender}\label{subsection:defender}
The defender's action is to implement the function $\mathsf{GetHardness}()$. The implementation must be efficiently computable, and the function must be chosen subject to maximum workload constraints on the authentication server. Otherwise, the optimal solution would simply be to set the cost parameter $k$ for each password to be as large as possible. In addition, the server should guarantee that each password is granted with at least some level of protection so that it will not make weak passwords weaker.

In an idealized setting where the defender knows the user password distribution we can implement the  function $\mathsf{GetHardness}(pw_u)$ as follows:  the authentication server first partitions all passwords into $\tau$ mutually exclusive groups $G_i$ with $i\in\{1,\cdots,\tau\}$ such that $\mathcal{P} = \bigcup_{i = 1}^\tau G_i$ and $\Pr[pw] > \Pr[pw']$ for every $pw \in G_i$ and $pw' \in G_{i+1}$. Here, $G_1$ will \revision{correspond} to the weakest group of passwords and $G_{\tau}$ corresponds to the group of strongest passwords. For each of the  $\lvert G_i\rvert$ passwords $pw \in G_i$ we assign the same hash cost parameter $k_i = \mathsf{GetHardness}(pw)$. 

\ignore{
  In practice, it may be overly optimistic to assume that the defender has perfect knowledge of the password distribution. Instead, we can partition passwords into groups $G_i$ based on estimates of the password strength. To estimate $\Pr[pw] $ we could use password strength meters or a differentially private count sketch which can provide a noisy estimate of the number of users who have selected the password $pw$. }

The cost of authenticating a password that is from $G_i$ is simply $k_i$. Therefore, the amortized server cost for verifying a correct password is: 
\begin{equation}\small
C_{SRV}=\sum_{i=1}^{\tau } k_i \cdot \Pr[pw \in G_i],
\end{equation}
where $\Pr[pw\in G_i] = \sum_{pw\in G_i}Pr[pw]$ is total probability mass of passwords in group $G_i$. In general, we will assume that the server has a maximum amortized cost $C_{max}$ that it is willing/able to incur for user authentication. Thus, the authentication server must pick the hash cost vector $\vec{k}=\{k_1,k_2,\cdots,k_{\tau}\}$ subject to the cost constraint $C_{SRV}\leq C_{max}$.   \revision{Additionally, we require that $k(pw_i) \geq k_{min}$ to ensure a minimum acceptable level of protection for all accounts.} The attacker will need to repeat this process independently for each user $u$. Thus, in our analysis we can focus on an individual user's account that the attacker is trying to crack. 
\vspace{-0.3cm}

\subsection{Action Space of Attacker}
After breaching the authentication server the attacker may run an offline dictionary attack. The attacker must fix an ordering $\pi$ over passwords $\mathcal{P}$ and a maximum number of guesses $B$ to check i.e., the attacker will check the first $B$ passwords in the ordering given by $\pi$.  If $B=0$ then the attacker gives up immediately without checking any passwords and if $B= \infty$ then the attacker will continue guessing until the password is cracked. The permutation $\pi$ specifies the order in which the attacker will guess passwords, i.e., the attacker will check password $pw_{\pi(1)}$ first then $pw_{\pi(2)}$ second etc... Thus, the tuple $(\pi,B)$ forms a \emph{strategy} of the adversary. Following that strategy 
the probability that the adversary succeeds in cracking a random user’s password is simply sum of probability of all passwords to be checked:
\begin{equation}\small
P_{ADV}=\lambda(\pi,B)=\sum_{i=1}^B p_{\pi(i)}\ .
\end{equation}
Here, we use short notation $p_{\pi(i)} = \Pr[pw_{\pi(i)}]$ which denotes the probability of the $i$th password in the ordering $\pi$. 

\subsection{Attacker’s Utility}
Given the estimated average value for one single password $v$ the expected gain of the attacker is simply $v \times \lambda(\pi,B)$ i.e., the probability that the password is cracked times the value $v$. Similarly, given a hash cost parameter vector $\vec{k}$ the expected cost of the attacker is 
$\sum^B_{i=1} k(pw_{\pi(i)})\cdot \left(1-\lambda(\pi,i-1)\right).$
We use the shorthand $k(pw) = k_i = \mathsf{GetHardness}(pw)$ for a password $pw \in G_i$. 
Intuitively, the probability that the first $i-1$ guesses are incorrect is  $\left(1-\lambda(\pi,i-1)\right)$ and we incur cost $k(pw_{\pi(i)})$ for the $i$'th guess if and only if the first $i-1$ guesses are incorrect. Note that $\lambda(\pi,0)=0$ so the attacker always pays cost $k(pw_{\pi(1)})$ for the first guess.
The adversary’s expected utility is the difference of expected gain and expected cost:
\begin{equation}\small
\begin{aligned}
&U_{ADV}\left(v,\vec{k},(\pi,B)\right)=v\cdot \lambda(\pi,B)-\sum^B_{i=1} k(pw_{\pi(i)})\cdot \left(1-\lambda(\pi,i-1)\right).
\end{aligned}
\end{equation}
\vspace{-0.5cm}

\subsection{Defender’s Utility} After the defender (leader) moves the offline attacker (follower) will respond with his/her utility optimizing strategy. We let $P_{ADV}^* $ denote the probability that the attacker cracks a random user's password when playing his/her optimal strategy. 
\begin{equation}\small
P_{ADV}^* = \lambda(\pi^*,B^*)\ ,  ~~~\mbox{where~~~} (\pi^*,B^*)=\arg \max_{\pi, B} U_{ADV}\left(v,\vec{k},(\pi,B)\right).
\end{equation}
$P_{ADV}^*$ will depend on the attacker's utility optimizing strategy which will in turn depend on value $v$ for a cracked password, the chosen cost parameters $k_i$ for each group $G_i$, and the user password distribution. Thus, we can define the authentication server’s utility as
\begin{equation}\small
U_{SRV}(\vec{k},v)=-P_{ADV}^* \ .
\end{equation} 

The objective of the authentication is to minimize the success rate $P_{ADV}^*(v,\vec{k})$ of the attacker by finding the optimal action i.e., a good way of  partitioning passwords into groups and selecting the optimal hash cost vector $\vec{k}$.  
Since the parameter $\vec{k}$ controls the cost of the hash function in passwords storage and authentication, we should increase $k_i$ for a specific group $G_i$ of passwords only if this is necessary to help deter the attacker from cracking passwords in this group $G_i$. The defender may not want to waste too much resource in protecting the weakest group $G_1$ of passwords when password value is high because they will be cracked easily regardless of the hash cost $k_1$.
\vspace{-0.3cm}
\subsection{Stackelberg Game Stages}
Since adversary’s utility depends on $(\pi,B)$ and $\vec{k}$, wherein $(\pi,B)$ is the responses to server’s predetermined hash cost vector $\vec{k}$. On the other hand, when server selects different hash cost parameter for different groups of password, it has to take the reaction of potential attackers into account.
Therefore, the interaction between the authentication server and the adversary can be modeled as a two stage Stackelberg Game. Then the problem of finding the optimal hash cost vector is reduced to the problem of computing the equilibrium of  Stackelberg game.

\ignore{
We will assume that the value $v$ is available to the defender before the game begins. For example, the defender might estimate $v$ from black market reports~\cite{CCS:Allodi17,goldForSilver,stockley2016} \footnote{We discuss the issue of robustness to inexact estimations of $v$ in \ref{sec:empirical}. Briefly, we find that the optimal vector $\vec{k}$ is (usually) not sensitive to moderate changes in the value $v$ which means that our mechanism will be expected to work well even when we only have loose estimates of $v$.} Similarly, we assume that the attacker know $\Pr[pwd_i]$ for each password $pwd_i\in \mathcal{P}$ and the attacker will learn the vector $\vec{k}$ of hash cost parameters after the server is breached as well as the particular groups $G_1,\ldots, G_\tau$ (typically represented implicitly).
}

In the Stackelberg game, the authentication server (leader) moves first (stage I); then the adversary follows (stage II). In stage I, the authentication server commits hash cost vector $\vec{k}=\{k_1,\cdots k_{\tau}\}$ for all groups of  passwords; in stage II, the adversary yields the optimal strategy $(\pi,B)$ for cracking a random user’s password. 
Through the interaction between the legitimate authentication server and the untargeted adversary who runs an offline attack, there will emerge an equilibrium in which no player in the game has the incentive to unilaterally change its strategy.  Thus, an equilibrium strategy profile $\left\{\vec{k}^*,(\pi^*,B^*)\right\}$ must satisfy
\begin{equation}\small
\begin{cases}
    U_{SRV}\left(\vec{k}^*,v\right)\geq U_{SRV}\left(\vec{k},v\right),& \forall \vec{k} \in \mathcal{F}_{C_{max}} ,\\
    U_{ADV}\left(v,\vec{k}^*,(\pi^*,B^*)\right)\geq U_{ADV}\left(v,\vec{k}^*,(\pi,B)\right),  &\forall(\pi,B)
\end{cases}
\end{equation}
Assuming that the grouping $G_1,\ldots,G_\tau$ of passwords is fixed. The computation of equilibrium strategy profile can be transformed to solve the following optimization problem, where $\Pr(pw_i)$, $G_1,
\cdots, G_{\tau}$, $C_{max}$ are input parameters and  $(\pi^*,B^*)$ and $\vec{k}^*$ 
are variables.
\begin{equation}\small
\label{eq:utility}
\begin{aligned}
\min_{\vec{k}^*, \pi^*, B*}& \lambda(\pi^*,B^*)\\
 \textrm{s.t.} \quad
& U_{ADV}\left(v,\vec{k},(\pi^*,B^*)\right) \geq U_{ADV}\left(v,\vec{k},(\pi,B)\right),~~\forall (\pi,B), \\
&\sum_{i=1}^{\tau } k_i \cdot \Pr[pw \in G_i] \leq C_{max},\\
& k_i \geq k_{min},   \mbox{~$\forall i \leq \tau$}.
\end{aligned}
\end{equation}

The solution of the above optimization problem is the equilibrium of our Stackelberg game. The first constraint implies that adversary will play his/her utility optimizing strategy i.e., given that the defender's action $\vec{k}^*$ is fixed the utility of the strategy $(\pi^*,B^*)$ is at least as large as any other strategy the attacker might follow. Thus, a rational attacker will check the first $B^*$ passwords in the order indicated by $\pi^*$ and then stop cracking passwords. The second constraint is due to resource limitations of authentication server. The third constraint sets lower-bound for the protection level. In order to tackle the first constraint, we need to specify the optimal checking sequence and the optimal number of passwords to be checked. 
\vspace{-0.3cm}

\section{Attacker and Defender Strategies} \label{sec:Analysis}
In the first subsection, we give an efficient algorithm to compute the attacker’s optimal strategy $(\pi^*,B^*)$ given the parameters $v$ and $\vec{k}$. This algorithm in turn is an important subroutine in our algorithm to find the best stragety $\vec{k}^*$ for the defender. 
\vspace{-0.3cm}
\subsection{Adversary’s Best Response (Greedy)}
In this section we show that the attacker's optimal ordering $\pi^*$ can be obtained by sorting passwords by their ``bang-for-buck'' ratio.  \revision{In particular, fixing an ordering $\pi$} we define the ratio $r_{\pi(i)}=\frac{p_{\pi(i)}}{k(pw_{\pi(i)})}$ which can be viewed as the priority of checking password $pw_{\pi(i)}$ i.e., the cost will be $k(pw_{\pi(i)})$ and the probability the password is correct is $p_{\pi(i)}$. \revision{Intuitively, the attacker's optimal strategy is to order passwords by their ``bang-for-buck'' ratio guessing passwords with higher checking priority first. Theorem \ref{thm:noinversions} formalizes this intuition by proving that the optimal checking sequence $\pi^*$ has no inversions.  }

We say a checking sequence $\pi$ has an \emph{inversion} with respect to $\vec{k}$ \revision{ if for some pair $a > b$ we have $r_{\pi(a)}>r_{\pi(b)}$ i.e., $pw_{\pi(b)}$ is scheduled to be checked before $pw_{\pi(a)}$ even though password $pw_{\pi(a)}$ has a higher ``bang-for-buck'' ratio. Recall that $pw_{\pi(b)}$ is the $b$'th password checked in the ordering $\pi$. }  \revision{The proof of Theorem \ref{thm:noinversions} can be found in the appendix \ref{app:proof}. Intuitively, we argue that consecutive inversions can always be swapped without decreasing the attacker's utility.}

\newcommand{\thmnoinversions}{\revision{Let $(\pi^*,B^*)$ denote the attacker's optimal strategy with respect to hash cost parameters $\vec{k}$ and let $\pi$ be an ordering with no inversions relative to $\vec{k}$ then  \[ U_{ADV}\left(v,\vec{k},(\pi,B^*)\right) \geq U_{ADV}\left(v,\vec{k},(\pi^*,B^*)\right) \ . \]}}
\begin{theorem}
\label{thm:noinversions}
\thmnoinversions
\end{theorem}
\vspace{-0.4cm}
Theorem \ref{thm:noinversions} gives us {\em an easy way to compute} the attacker's optimal ordering $\pi^*$ over passwords \revision{i.e.,  by sorting passwords according to their ``bang-for-buck'' ratio.} \revision{It remains to find the attacker's optimal guessing budget $B^*$}. \revision{As we previously mentioned the password distributions we consider can be compressed by grouping passwords with equal probability into equivalence sets.} \revision{Once we have our cost vector $\vec{k}$ and have implemented \hard{}  we can further partition password equivalence sets such that passwords in each set additionally have the same bang-for-buck ratio.
 Theorem \ref{corollary} tells us that the optimal attacker strategy will either guess {\em all} of the passwords in such an equivalence set $ec_j$ or {\em none} of them. Thus, when we search for $B^*$ we only need to consider $n^{\prime}+1$ possible values of this parameter. } We will use this observation to improve the efficiency of our algorithm to compute the optimal attacker strategy. 

\newcommand{\thmcompact}{Let $(\pi^*,B^*)$ be the optimal strategy of \revision{the} adversary and given two passwords $pw_i$ and $pw_j$ in the same equivalence set.  Then
\begin{equation}
\label{eq:theorem2}
\mathsf{Inv}_{\pi^*}(i)\leq B^* \Leftrightarrow \mathsf{Inv}_{\pi^*}(j)\leq B^* \ . 
\end{equation}}

\newcommand{\maincorollary}{Let $(\pi^*,B^*)$ denote the attacker's optimal strategy with respect to hash cost parameters $\vec{k}$. Suppose that passwords can be partitioned into $n$ equivalence sets $es_1,\ldots, es_{n^{\prime}}$ such that passwords $pw_a, pw_b \in es_i$ have the same probability and hash cost i.e., $p_a=p_b = p^i$ and $k(pw_a) = k(pw_b)= k^i$. Let $r^i = p^i/k^i$ denote the bang-for-buck ratio of equivalence set $es_i$ and assume that $r^1 \geq r^2 \geq \ldots \geq r_{n^{\prime}}$ then  $B^* \in \left\{0, |es_1| ,|es_1|+|es_2|,\cdots ,\sum_{i=1}^{n^{\prime}}|es_i|\right\}$.}
\begin{theorem}\label{corollary}
\revision{\maincorollary }
\end{theorem}
The proof of both theorems can be found in Appendix \ref{app:proof}.  Theorem \ref{corollary} implies that when cracking users’ accounts the adversary increases number of guesses $B$ by the size of the next equivalence set (if there is net profit by doing so). Therefore, the attacker finds the optimal strategy $(\pi^*, B^*)$  with Algorithm $\mathsf{BestRes}(v, \vec{k}, D)$ in time $\mathcal{O}(n^{\prime}\log n^{\prime})$ \revision{ --- see Algorithm \ref{alg:response} in Appendix \ref{app:algorithm}. The running time is dominated by the cost of sorting our $n^{\prime}$ equivalence sets. }
\vspace{-0.3cm}

\subsection{The Optimal Strategy of Selecting Hash Cost Vector} \label{subsec:OptimizingCostVector}
\revision{In the previous section we showed that there is an efficient greedy algorithm $\mathsf{BestRes}(v, \vec{k}, D)$ which takes as input a cost vector $\vec{k}$, a value $v$ and a (compressed) description of the password distribution $D$ computes the the  attacker's best response $(\pi^*, B^*)$ and outputs $\lambda(\pi^*, B^*)$ --- the fraction of cracked passwords. Using this algorithm $\mathsf{BestRes}(v, \vec{k}, D)$ as a blackbox we can apply derivative-free optimization to  the optimization problem in equation (\ref{eq:utility}) to find a good hash cost vector $\vec{k}$ which minimizes the objective $\lambda(\pi^*, B^*)$} There are many derivative-free optimization solvers available in the literature \cite{rios2013derivative}, generally they fall into two categorizes, deterministic algorithms (such as Nelder-Mead) and evolutionary algorithm (such as BITEOPT\cite{biteopt} and CMA-EA). \revision{ We refer our solver to as $\mathsf{OptHashCostVec}(v,C_{max}, k_{min}, D)$. The algorithm takes as input the parameters of the optimization problem (i.e., password value $v$, $C_{max}$, $k_{min}$, and a (compressed) description of the password distribution $D$) and outputs an optimized hash cost vector $\vec{k}$. }

 During each iteration of  $\mathsf{OptHashCostVec}(\cdot)$, some candidates $\{\vec{k}_{c_i}\}$ are proposed,  together they are referred as \emph{population}. For each candidate solution $\vec{k}_{c_i}$  we use our greedy algorithm $\mathsf{BestRes}(v, \vec{k}_{c_i}, D)$ to compute the attacker's best response $(\pi^*, B^*)$ \revision{  i.e., fixing any feasible cost vector $\vec{k}_{c_i}$ we can compute the corresponding value of the objective function  $P_{adv,\vec{k}_{c_i}} := \sum_{i=1}^{B^*}p_{\pi^{*}(i)}$.} We record the corresponding success rate $P_{adv,\vec{k}_{c_i}}$ of the attacker as ``fitness''.  At the end of each iteration, the population is updated according to fitness of its' members, the update could be either through deterministic transformation (Nelder-Mead) or randomized evolution (BITEOPT, CMA-EA). When the iteration number reaches a pre-defined value $ite$, the best fit member $\vec{k}^*$ and its fitness $P_{adv}^*$ are returned.
\vspace{-0.3cm}

\section{Empirical Analysis} \label{sec:empirical}

\revision{In this section, we design experiments to analyze the effectiveness of DAHash. At a high level we first fix (compressed) password distributions $D_{train}$ and $D_{eval}$ based on empirical password datasets and an implementation of \hard{}. Fixing the DAHash parameters $v$, $C_{max}$ and $k_{min}$ we use our algorithm $\mathsf{OptHashCostVec}(v,C_{max}, k_{min}, D_{train})$ to optimize the cost vector $\vec{k}^*$ and then we compute the attacker's optimal response $\mathsf{BestRes}(v,\vec{k}^*, D_{eval})$.  By setting $D_{train} = D_{eval}$ we can model the idealized scenario where the defender has perfect knowledge of the password distribution. Similarly, by setting $D_{train} \neq D_{eval}$ we can model the performance of DAHash when the defender optimizes $\vec{k}^*$ without perfect knowledge of the password distribution. In each experiment we fix $k_{min} = C_{max}/10$ and we plot the fraction of cracked passwords as the value to cost ratio $v/C_{max}$ varies. We compare DAHash with traditional password hashing fixing the hash cost to be $C_{max}$ for every password to ensure that the amortized server workload is equivalent. Before presenting our results we first describe how we define the password distributions $D_{train}$ and $D_{eval}$ and how we implement \hard{}. }

\subsection{The Password Distribution}
\revision{
One of the challenges in evaluating DAHash is that the exact distribution over user passwords is unkown. However, there are many empirical password datasets available due to password breaches. We describe two methods for deriving password distributions from password datasets. 

\subsubsection{Empirical Password Datasets} We consider nine empirical password datasets (along with their size $N$): Bfield ($0.54$ million), Brazzers ($0.93$ million), Clixsense ($2.2$ million), CSDN ($6.4$ million), LinkedIn ($174$ million), Neopets ($68.3$ million), RockYou ($32.6$ million), 000webhost ($153$ million) and Yahoo! ($69.3$ million). Plaintext passwords are available for all datasets except for the differentially private LinkedIn ~\cite{CS:HMBSD20} and Yahoo!~\cite{SP:Bonneau12,NDSS:BloDatBon16} frequency corpuses which intentionally omit passwords. With the exception of the Yahoo! frequency corpus all of the datasets are derived from password breaches. The differentially LinkedIn dataset is derived from cracked LinkedIn passwords \footnote{The LinkedIn password is derived from 174 million (out of 177.5 million) cracked password hashes which were cracked by KoreLogic~\cite{CS:HMBSD20}. Thus, the dataset omits $2\%$ of uncracked passwords. Another caveat is that the LinkedIn dataset only contains $164.6$ million unique e-mail addresses so there are some e-mail addresses with multiple associated password hashes.}. Formally, given $N$ user accounts $u_1,\ldots, u_N$ a dataset of passwords is a list $D = pw_{u_1},\ldots,pw_{u_N} \in \mathcal{P}$ of  passwords each user selected. We can view each of these passwords $pw_{u_i}$ as being sampled from some unkown distribution $D_{real}$. 

\vspace{-0.5cm}
\subsubsection{Empirical Distribution.}  Given a dataset of $N$ user passwords the corresponding password frequency list is simply a list of numbers $f_1 \geq f_2 \geq \ldots $ where $f_i$ is the number of users who selected the $i$th most popular password in the dataset --- note that $\sum_{i} f_i = N$. In the empirical password distribution we define the probability of the $i$th most likely password to be $\hat{p}_i= f_i/N$. In our experiments using the empirical password distribution we will set $D_{train}=D_{eval}$ i.e., we assume that the empirical password distribution is the real password distribution and that the defender knows this distribution. 

In our experiments we implement \hard{} by partitioning the password dataset $D_{train}$ into $\tau$ groups $G_1,\ldots, G_\tau$ using $\tau-1$ frequency thresholds $t_1 > \ldots > t_{\tau-1}$ i.e., $G_1 = \{i:f_i  \geq t_1\}$,  $G_{j} = \{i: t_{j-1} >  f_i \geq t_j \} $ for $1< j < \tau$ and $G_\tau = \{i: f_i < t_{\tau-1}\}$. Fixing a hash cost vector $\vec{k}=(k_1,\ldots, k_{\tau})$ we will assign passwords in group $G_j$ to have cost $k_j$ i.e., \hard{pw}$=k_j$ for $pw \in G_j$. We pick the thresholds to ensure that the probability mass $Pr[G_j] = \sum_{i \in G_j} f_i/N$ of each group is approximately balanced (without separating passwords in an equivalence set). While there are certainly other ways that \hard{} could be implemented (e.g., balancing number of passwords/equivalence sets in each group) we found that balancing the probability mass was most effective. }
 
\revision{ \noindent {\bf Good-Turing Frequency Estimation.} 
One disadvantage of using the empirical distribution is that it can often overestimate the success rate of an adversary. For example, let $\hat{\lambda}_B:= \sum_{i=1}^B \hat{p}_i$ and $N^{\prime} \leq N$ denote the number of distinct passwords in our dataset then we will {\em always } have $\hat{\lambda}_{N^{\prime}}:=\sum_{i\leq N^{\prime}} \hat{p}_i = 1$ which is inaccurate whenever $N \leq \left|\mathcal{P}\right|$. However, when $B \ll N$ we will have $\hat{\lambda}_B \approx \lambda_B$ i.e., the empirical distribution will closely match the real distribution. Thus, we will use the empirical distribution to evaluate the performance of DAHash when the value to cost ratio $v/C_{max}$ is smaller (e.g, $v/C_{max} \ll 10^8$) and we will highlight uncertain regions of the curve using Good-Turing frequency estimation. 

Let $N_f = \vert\{i : f_i=f \}\vert$ denote number of distinct passwords in our dataset that occur exactly $f$ times and let $B_f= \sum_{i > f} N_i$ denote the number of distinct passwords that occur more than $f$ times. Finally, let $E_f := |\lambda_{B_f} - \hat{\lambda}_{N_{B_f}}|$ denote the error of our estimate for $\lambda_{B_{ f}}$, the total probability of the top $B_{f}$ passwords in the real distribution. If our dataset consists of $N$ independent samples from an unknown distribution then Good-Turing frequency estimation tells us that the total probability mass of all passwords that appear exactly $f$ times is approximately $U_f:=(f+1)N_{f+1}/N$ e.g., the total probability mass of unseen passwords  is $U_0 = N_1/N$. This would imply that ${\lambda}_{B_f} \geq 1 - \sum_{j=0}^f U_j =  1-\sum_{j=0}^i \frac{(j+1)N_{j+1}}{N}$ and $E_f  \leq U_f$.

The following table plots our error upper bound $U_f$ for $0 \leq f \leq 10$ for 9 datasets. Fixing a target error threshold $\epsilon$ we define $f_{\epsilon} = \min\{i: U_i \leq \epsilon\}$ i.e., the minimum index such that the error is smaller than $\epsilon$. In our experiments we focus on error thresholds $\epsilon \in \{0.1, 0.01\}$. For example, for the Yahoo! (resp. Bfield) dataset we have $f_{0.1} = 1$ (resp. $j_{0.1}=2$) and $j_{0.01}=6$ (resp. $j_{0.01}=5$). As soon as we see passwords with frequency {\em at most} $j_{0.1}$ (resp. $j_{0.01}$) start to get cracked we highlight the points on our plots with a red (resp. yellow).  }
\vspace{-0.5cm}
\begin{table}[htb]
\caption{Error Upper Bounds: $U_i$ for Different Password Datasets}
\begin{center}
\begin{tabular}{cccccccccc}
\hline
       & Bfield & Brazzers & Clixsense & CSDN  & Linkedin & Neopets & Rockyou & 000webhost & Yahoo! \\ \hline
$U_0$  & \red{0.69}   & \red{0.531}    & \red{0.655}     & \red{0.557} & \red{0.123}    & \red{0.315}   & \red{0.365}   & \red{0.59}       & \red{0.425} \\ \hline
$U_1$  & \red{0.101}  & \red{0.126}    & \red{0.095}     & \red{0.092} &\red{0.321}    & \red{0.093}   & \red{0.081}   & \red{0.124}      & \red{0.065} \\ \hline
$U_2$  & \red{0.036}  & \red{0.054}    &\yellow{0.038}     & \yellow{0.034} & \red{0.043}    & \yellow{0.051}   & \yellow{0.036}   & \red{0.055}      & \yellow{0.031} \\ \hline
$U_3$  & \yellow{0.02}   & \yellow{0.03}     & \yellow{0.023}     & \yellow{0.018} & \yellow{0.055}    &\yellow{0.034}   & \yellow{0.022}   & \yellow{0.034}      & \yellow{0.021} \\ \hline
$U_4$  & \yellow{0.014}  & \yellow{0.02}     & \yellow{0.016}     & \yellow{0.012} & \yellow{0.018}    & \yellow{0.025}   & \yellow{0.017}   & \yellow{0.022}      & \yellow{0.015} \\ \hline
$U_5$  & \yellow{0.01}   & \yellow{0.014}    & \yellow{0.011}     & \yellow{0.008} & \yellow{0.021}    & \yellow{0.02}    & \yellow{0.013}   & \yellow{0.016}      & \yellow{0.012} \\ \hline
$U_6$  & 0.008  & \yellow{0.011}    &\yellow{0.009}     & 0.006 & \yellow{0.011}    & \yellow{0.016}   & \yellow{0.011}   & \yellow{0.012}      & \yellow{0.01}  \\ \hline
$U_7$  & 0.007  & \yellow{0.01}     & 0.007     & 0.005 & \yellow{0.011}    & \yellow{0.013}   & \yellow{0.01}    & \yellow{0.009}      & 0.009 \\ \hline
$U_8$  & 0.006  & 0.008    & 0.006     & 0.004 & \yellow{0.008}    & \yellow{0.011}   & 0.009   & 0.008      & 0.008 \\ \hline
$U_9$  & 0.005  & 0.007    & 0.005     & 0.004 & 0.007    & \yellow{0.01}    & 0.008   & 0.006      & 0.007 \\ \hline
$U_{10}$ & 0.004  & 0.007    & 0.004     & 0.003 & 0.006    & 0.009   & 0.007   & 0.005      & 0.006 \\ \hline
\end{tabular}
\end{center}
\label{default}
\end{table}
\vspace{-1cm}
\revision{ 
\subsubsection{Monte Carlo Distribution} As we observed previously the empirical password distribution can be highly inaccurate when $v/C_{max}$ is large. Thus, we use a different approach to evaluate the performance of DAHash when $v/C_{max}$ is large. In particular, we subsample passwords, obtain gussing numbers for each of these passwords and fit our distribution to the corresponding guessing curve. We follow the following procedure to derive a distribution: (1) subsample $s$ passwords $D_s$ from dataset $D$ with replacement; (2) for each subsampled passwords $pw \in D_s$ we use the Password Guessing Service~\cite{USENIX:USBCCKKMMS15} to obtain a guessing number $\guessing$ which uses Monte Carlo methods~\cite{CCS:DelFil15} to estimate how many guesses an attacker would need to crack $pw$ \footnote{The Password Guessing Service~\cite{USENIX:USBCCKKMMS15} gives multiple different guessing numbers for each password based on different sophisticated cracking models e.g., Markov, PCFG, Neural Networks. We follow the suggestion of the authors ~\cite{USENIX:USBCCKKMMS15} and use the minimum guessing number (over all autmated approached) as our final estimate.}. (3) For each $i \leq 199$ we fix guessing thresholds $t_0 <  t_1 < \ldots < t_{199}$ with $t_0:=0$, $t_1:=15$, $t_i - t_{i-1} = 1.15^{i+25}$, and $t_{199} = \max_{pw \in D_s} \{\guessing\}$. (4) For each $i \leq 199$ we compute $g_i$, the number of samples $pw \in D_s$ with $\guessing \in[t_{i-1},t_i)$. (5) We output a compressed distribution with $200$ equivalences sets using histogram density i.e., the $i$th equivalence set contains $t_{i}-t_{i-1}$ passwords each with probability $\frac{g_i}{s \times (t_i-t_{i-1})}$. 

In our experiments we repeat this process twice with $s=12,500$ subsamples to obtain two password distributions $D_{train}$ and $D_{eval}$. One advantage of this approach is that it allows us to evaluate the performance of DAHash against a state of the art password cracker when the ratio $v/C_{max}$ is large. The disadvantage is that the distributions $D_{train}$ and $D_{eval}$ we extract are based on {\em current} state of the art password cracking models. It is possible that we optimized our DAHash parameters with respect to the wrong distribution if an attacker develops an improved password cracking model in the future.
}

\revision{
{\noindent \bf Implementing \hard{} for Monte Carlo Distributions.} For Monte Carlo distribution \hard{pw} depends on the guessing number $\guessing$. In particular, we fix thresholds points $x_{1} > \ldots >  x_{\tau-1}$ and (implicitly) partition passwords into $\tau$ groups $G_1,\ldots, G_t$ using these thresholds i.e., $G_i = \{ pw~:~ x_{i-1}  \geq \guessing > x_{i}\}$. Thus, \hard{pw} would compute $\guessing$ and assign hash cost $k_i$ if $pw \in G_i$. As before the thresholds $x_1,\ldots, x_{\tau-1}$ are selected to (approximately) balance the probability mass in each group. }

\tikzDefaults
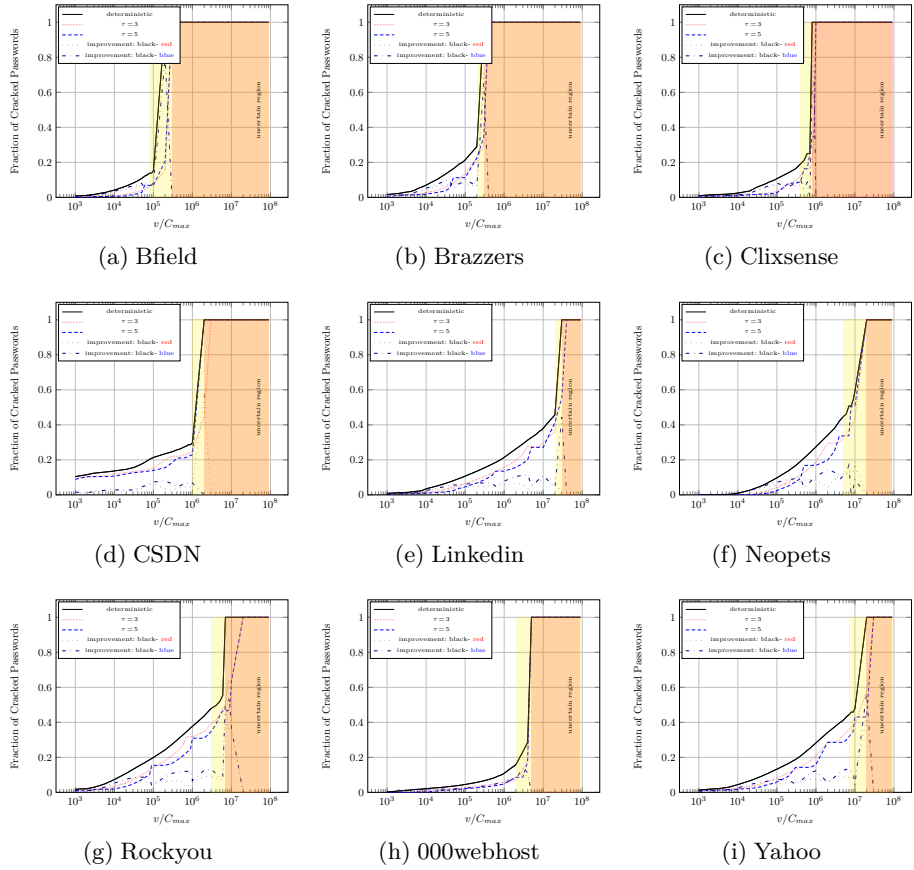
\begin{figure*}[ht]\centering
\subfloat[Bfield]{
\begin{tikzpicture}[scale=0.45]
\begin{semilogxaxis}[ymin=0]

\addplot+[stack plots=y] file {./empirical/benchmark/bfield.dat};

\addplot file {./empirical/tau3/bfield.dat};
\addplot file {./empirical/tau5/bfield.dat};
\addplot+[stack plots=y, stack dir=minus] file {./empirical/tau3/bfield.dat};
\addplot[stack plots=y, stack dir=plus, forget plot] file {./empirical/tau3/bfield.dat};
\addplot+[stack plots=y, stack dir=minus] file {./empirical/tau5/bfield.dat};

\path[name path=B] (axis cs:(9 * 1e7,0) -- (axis cs:(9 * 1e7,1);
\path[name path=A] (axis cs:(300000,0) -- (axis cs:(300000,1);
\tikzfillbetween[of=A and B, on layer=main]{red, opacity=0.2};

\path[name path=A1] (axis cs:(80000,0) -- (axis cs:(80000,1);
\tikzfillbetween[of=A1 and B, on layer=main]{yellow, opacity=0.2};

\node at (axis cs:5*1e7,0.5) {\rotatebox{90}{\tiny uncertain region}};
\end{semilogxaxis}
\end{tikzpicture}
\label{fig:bfield}
}
\hfill
\subfloat[Brazzers]{
\begin{tikzpicture}[scale=0.45]
\begin{semilogxaxis}[ymin=0]
\addplot+[stack plots=y] file {./empirical/benchmark/brazzers.dat};
\addplot file {./empirical/tau3/brazzers.dat};
\addplot file {./empirical/tau5/brazzers.dat};
\addplot+[stack plots=y, stack dir=minus] file {./empirical/tau3/brazzers.dat};
\addplot[stack plots=y, stack dir=plus, forget plot] file {./empirical/tau3/brazzers.dat};
\addplot+[stack plots=y, stack dir=minus] file {./empirical/tau5/brazzers.dat};

\path[name path=B] (axis cs:(9 * 1e7,0) -- (axis cs:(9 * 1e7,1);
\path[name path=A1] (axis cs:(3 *1e5,0) -- (axis cs:(3 * 1e5,1);
\tikzfillbetween[of=A1 and B, on layer=main]{red, opacity=0.2};

\path[name path=A] (axis cs:(2 *1e5,0) -- (axis cs:(2 * 1e5,1);
\tikzfillbetween[of=A and B, on layer=main]{yellow, opacity=0.2};

\node at (axis cs:5*1e7,0.5) {\rotatebox{90}{\tiny uncertain region}};

\end{semilogxaxis}
\end{tikzpicture}
\label{fig:brazzers}
}
\hfill
\subfloat[Clixsense]{
\begin{tikzpicture}[scale=0.45]
\begin{semilogxaxis}[ymin=0]
\addplot+[stack plots=y] file {./empirical/benchmark/clixsense.dat};
\addplot file {./empirical/tau3/clixsense.dat};
\addplot file {./empirical/tau5/clixsense.dat};
\addplot+[stack plots=y, stack dir=minus] file {./empirical/tau3/clixsense.dat};
\addplot[stack plots=y, stack dir=plus, forget plot] file {./empirical/tau3/clixsense.dat};
\addplot+[stack plots=y, stack dir=minus] file {./empirical/tau5/clixsense.dat};

\path[name path=A] (axis cs:(4 * 1e5,0) -- (axis cs:(4 * 1e5,1);
\path[name path=B] (axis cs:(9 * 1e7,0) -- (axis cs:(9 * 1e7,1);
\tikzfillbetween[of=A and B, on layer=main]{yellow, opacity=0.2};

\path[name path=A1] (axis cs:(8 * 1e5,0) -- (axis cs:(8 * 1e5,1);
\tikzfillbetween[of=A1 and B, on layer=main]{red, opacity=0.2};
\node at (axis cs:5*1e7,0.5) {\rotatebox{90}{\tiny uncertain region}};

\end{semilogxaxis}
\end{tikzpicture}
\label{fig:clixsense}
}

\subfloat[CSDN]{
\begin{tikzpicture}[scale=0.45]
\begin{semilogxaxis}[ymin=0]
\addplot+[stack plots=y] file {./empirical/benchmark/csdn.dat};
\addplot file {./empirical/tau3/csdn.dat};
\addplot file {./empirical/tau5/csdn.dat};
\addplot+[stack plots=y, stack dir=minus] file {./empirical/tau3/csdn.dat};
\addplot[stack plots=y, stack dir=plus, forget plot] file {./empirical/tau3/csdn.dat};
\addplot+[stack plots=y, stack dir=minus] file {./empirical/tau5/csdn.dat};

\path[name path=A] (axis cs:(2* 1e6,0) -- (axis cs:(2* 1e6,1);
\path[name path=B] (axis cs:(9 * 1e7,0) -- (axis cs:(9 * 1e7,1);
\tikzfillbetween[of=A and B, on layer=main]{red, opacity=0.2};

\path[name path=A1] (axis cs:(1* 1e6,0) -- (axis cs:(1* 1e6,1);
\tikzfillbetween[of=A1 and B, on layer=main]{yellow, opacity=0.2};
\node at (axis cs:5*1e7,0.5) {\rotatebox{90}{\tiny uncertain region}};

\end{semilogxaxis}
\end{tikzpicture}
\label{fig:csdn}
}
\hfill
\subfloat[Linkedin]{
\begin{tikzpicture}[scale=0.45]
\begin{semilogxaxis}[ymin=0]
\addplot+[stack plots=y] file {./empirical/benchmark/linkedin.dat};
\addplot file {./empirical/tau3/linkedin.dat};
\addplot file {./empirical/tau5/linkedin.dat};
\addplot+[stack plots=y, stack dir=minus] file {./empirical/tau3/linkedin.dat};
\addplot[stack plots=y, stack dir=plus, forget plot] file {./empirical/tau3/linkedin.dat};
\addplot+[stack plots=y, stack dir=minus] file {./empirical/tau5/linkedin.dat};

\path[name path=A] (axis cs:(3* 1e7,0) -- (axis cs:(3* 1e7,1);
\path[name path=B] (axis cs:(9 * 1e7,0) -- (axis cs:(9 * 1e7,1);
\tikzfillbetween[of=A and B, on layer=main]{red, opacity=0.2};

\path[name path=A1] (axis cs:(2* 1e7,0) -- (axis cs:(2* 1e7,1);
\tikzfillbetween[of=A1 and B, on layer=main]{yellow, opacity=0.2};

\node at (axis cs:5*1e7,0.5) {\rotatebox{90}{\tiny uncertain region}};

\end{semilogxaxis}
\end{tikzpicture}
\label{fig:linkedin}
}
\hfill
\subfloat[Neopets]{
\begin{tikzpicture}[scale=0.45]
\begin{semilogxaxis}[ymin=0]
\addplot+[stack plots=y] file {./empirical/benchmark/neopets.dat};
\addplot file {./empirical/tau3/neopets.dat};
\addplot file {./empirical/tau5/neopets.dat};
\addplot+[stack plots=y, stack dir=minus] file {./empirical/tau3/neopets.dat};
\addplot[stack plots=y, stack dir=plus, forget plot] file {./empirical/tau3/neopets.dat};
\addplot+[stack plots=y, stack dir=minus] file {./empirical/tau5/neopets.dat};

\path[name path=A] (axis cs:(2* 1e7,0) -- (axis cs:(2* 1e7,1);
\path[name path=B] (axis cs:(9 * 1e7,0) -- (axis cs:(9 * 1e7,1);
\tikzfillbetween[of=A and B, on layer=main]{red, opacity=0.2};

\path[name path=A1] (axis cs:(5* 1e6,0) -- (axis cs:(5* 1e6,1);
\tikzfillbetween[of=A1 and B, on layer=main]{yellow, opacity=0.2};

\node at (axis cs:5*1e7,0.5) {\rotatebox{90}{\tiny uncertain region}};

\end{semilogxaxis}
\end{tikzpicture}
\label{fig:neopets}
}

\subfloat[Rockyou]{
\begin{tikzpicture}[scale=0.45]
\begin{semilogxaxis}[ymin=0]
\addplot+[stack plots=y] file {./empirical/benchmark/rockyou.dat};
\addplot file {./empirical/tau3/rockyou.dat};
\addplot file {./empirical/tau5/rockyou.dat};
\addplot+[stack plots=y, stack dir=minus] file {./empirical/tau3/rockyou.dat};
\addplot[stack plots=y, stack dir=plus, forget plot] file {./empirical/tau3/rockyou.dat};
\addplot+[stack plots=y, stack dir=minus] file {./empirical/tau5/rockyou.dat};

\path[name path=A] (axis cs:(7 * 1e6,0) -- (axis cs:(7 * 1e6,1);
\path[name path=B] (axis cs:(9 * 1e7,0) -- (axis cs:(9 * 1e7,1);
\tikzfillbetween[of=A and B, on layer=main]{red, opacity=0.2};

\path[name path=A1] (axis cs:(3 * 1e6,0) -- (axis cs:(3 * 1e6,1);
\tikzfillbetween[of=A1 and B, on layer=main]{yellow, opacity=0.2};

\node at (axis cs:5*1e7,0.5) {\rotatebox{90}{\tiny uncertain region}};

\end{semilogxaxis}
\end{tikzpicture}
\label{fig:rockyou}
}
\hfill
\subfloat[000webhost]{
\begin{tikzpicture}[scale=0.45]
\begin{semilogxaxis}[ymin=0]
\addplot+[stack plots=y] file {./empirical/benchmark/webhost.dat};
\addplot file {./empirical/tau3/webhost.dat};
\addplot file {./empirical/tau5/webhost.dat};
\addplot+[stack plots=y, stack dir=minus] file {./empirical/tau3/webhost.dat};
\addplot[stack plots=y, stack dir=plus, forget plot] file {./empirical/tau3/webhost.dat};
\addplot+[stack plots=y, stack dir=minus] file {./empirical/tau5/webhost.dat};

\path[name path=A] (axis cs: 5* 1e6,0) -- (axis cs: 5* 1e6,1);
\path[name path=B] (axis cs: 9 * 1e7,0) -- (axis cs:9 * 1e7,1);
\tikzfillbetween[of=A and B, on layer=main]{red, opacity=0.2};

\path[name path=A1] (axis cs: 2* 1e6,0) -- (axis cs: 2* 1e6,1);
\tikzfillbetween[of=A1 and B, on layer=main]{yellow, opacity=0.2};

\node at (axis cs:5*1e7,0.5) {\rotatebox{90}{\tiny uncertain region}};

\end{semilogxaxis}
\end{tikzpicture}
\label{fig:webhost}
}
\hfill
\subfloat[Yahoo]{
\begin{tikzpicture}[scale=0.45]
\begin{semilogxaxis}[ymin=0]
\addplot+[stack plots=y] file {./empirical/benchmark/yahoo.dat};
\addplot file {./empirical/tau3/yahoo.dat};
\addplot file {./empirical/tau5/yahoo.dat};
\addplot+[stack plots=y, stack dir=minus] file {./empirical/tau3/yahoo.dat};
\addplot[stack plots=y, stack dir=plus, forget plot] file {./empirical/tau3/yahoo.dat};
\addplot+[stack plots=y, stack dir=minus] file {./empirical/tau5/yahoo.dat};

\path[name path=A] (axis cs:(2* 1e7,0) -- (axis cs:(2* 1e7,1);
\path[name path=B] (axis cs:(9 * 1e7,0) -- (axis cs:(9 * 1e7,1);
\tikzfillbetween[of=A and B, on layer=main]{red, opacity=0.2};

\path[name path=A1] (axis cs:(7* 1e6,0) -- (axis cs:(7* 1e6,1);
\tikzfillbetween[of=A1 and B, on layer=main]{yellow, opacity=0.2};

\node at (axis cs:5*1e7,0.5) {\rotatebox{90}{\tiny uncertain region}};

\end{semilogxaxis}
\end{tikzpicture}
\label{fig:Yahoo}
}
\caption{Adversary Success Rate vs $v/C_{max}$ for Empirical Distributions} {\par \small the red (resp. yellow) shaded areas denote unconfident regions where the the empirical distribution might diverges from the real distribution $U_i \geq 0.1$ (resp. $U_i \geq 0.01$).} 
\label{fig:empirical}
\vspace{-0.4cm}
\end{figure*}
\tikzDefaults
\begin{figure*}[ht]\centering
\subfloat[Bfield]
{
\begin{tikzpicture}[scale = 0.44]
\begin{semilogxaxis}[xmax= 1e12, ymin=0]

\addplot+[stack plots=y] file {./monte/benchmark/bfield.dat};
\addplot file {./monte/tau3/bfield.dat};
\addplot file {./monte/tau5/bfield.dat};
\addplot+[stack plots=y, stack dir=minus] file {./monte/tau3/bfield.dat};
\addplot[stack plots=y, stack dir=plus, forget plot] file {./monte/tau3/bfield.dat};
\addplot+[stack plots=y, stack dir=minus] file {./monte/tau5/bfield.dat};
\end{semilogxaxis}
\end{tikzpicture}
\label{fig:bfield}
}
\hfill
\subfloat[Brazzers]{
\begin{tikzpicture}[scale = 0.44]
\begin{semilogxaxis}[xmax= 1e12, ymin=0]
\addplot+[stack plots=y] file {./monte/benchmark/brazzers.dat};

\addplot file {./monte/tau3/brazzers.dat};
\addplot file {./monte/tau5/brazzers.dat};
\addplot+[stack plots=y, stack dir=minus] file {./monte/tau3/brazzers.dat};
\addplot[stack plots=y, stack dir=plus, forget plot] file {./monte/tau3/brazzers.dat};
\addplot+[stack plots=y, stack dir=minus] file {./monte/tau5/brazzers.dat};
\end{semilogxaxis}
\end{tikzpicture}
\label{fig:brazzers}
}
\hfill
\subfloat[Clixsense]{
\begin{tikzpicture}[scale = 0.44]
\begin{semilogxaxis}[ xmax= 1e12, ymin=0]
\addplot+[stack plots=y] file {./monte/benchmark/clixsense.dat};
\addplot file {./monte/tau3/clixsense.dat};
\addplot file {./monte/tau5/clixsense.dat};
\addplot+[stack plots=y, stack dir=minus] file {./monte/tau3/clixsense.dat};
\addplot[stack plots=y, stack dir=plus, forget plot] file {./monte/tau3/clixsense.dat};
\addplot+[stack plots=y, stack dir=minus] file {./monte/tau5/clixsense.dat};
\end{semilogxaxis}
\end{tikzpicture}
\label{fig:brazzers}
}

\subfloat[CSDN]{
\begin{tikzpicture}[scale = 0.44]
\begin{semilogxaxis}[ xmax= 1e12, ymin=0]
\addplot+[stack plots=y] file {./monte/benchmark/csdn.dat};

\addplot file {./monte/tau3/csdn.dat};
\addplot file {./monte/tau5/csdn.dat};
\addplot+[stack plots=y, stack dir=minus] file {./monte/tau3/csdn.dat};
\addplot[stack plots=y, stack dir=plus, forget plot] file {./monte/tau3/csdn.dat};
\addplot+[stack plots=y, stack dir=minus] file {./monte/tau5/csdn.dat};
\end{semilogxaxis}
\end{tikzpicture}
\label{fig:csdn}
}
\hfill
\subfloat[Neopets]{
\begin{tikzpicture}[scale = 0.44]
\begin{semilogxaxis}[ xmax= 1e12, ymin=0]
\addplot+[stack plots=y] file {./monte/benchmark/neopets.dat};

\addplot file {./monte/tau3/neopets.dat};
\addplot file {./monte/tau5/neopets.dat};
\addplot+[stack plots=y, stack dir=minus] file {./monte/tau3/neopets.dat};
\addplot[stack plots=y, stack dir=plus, forget plot] file {./monte/tau3/neopets.dat};
\addplot+[stack plots=y, stack dir=minus] file {./monte/tau5/neopets.dat};
\end{semilogxaxis}
\end{tikzpicture}
\label{fig:neopets}
}
\hfill
\subfloat[000webhost]{
\begin{tikzpicture}[scale = 0.44]
\begin{semilogxaxis}[ xmax= 1e12, ymin=0]
\addplot+[stack plots=y] file {./monte/benchmark/000webhost.dat};

\addplot file {./monte/tau3/000webhost.dat};
\addplot file {./monte/tau5/000webhost.dat};
\addplot+[stack plots=y, stack dir=minus] file {./monte/tau3/000webhost.dat};
\addplot[stack plots=y, stack dir=plus, forget plot] file {./monte/tau3/000webhost.dat};
\addplot+[stack plots=y, stack dir=minus] file {./monte/tau5/000webhost.dat};
\end{semilogxaxis}
\end{tikzpicture}
\label{fig:webhost}
}
\caption{Adversary Success Rate vs $v/C_{max}$ for Monte Carlo Distributions} 
\label{fig:monte}
\vspace{-0.5cm}
\end{figure*}
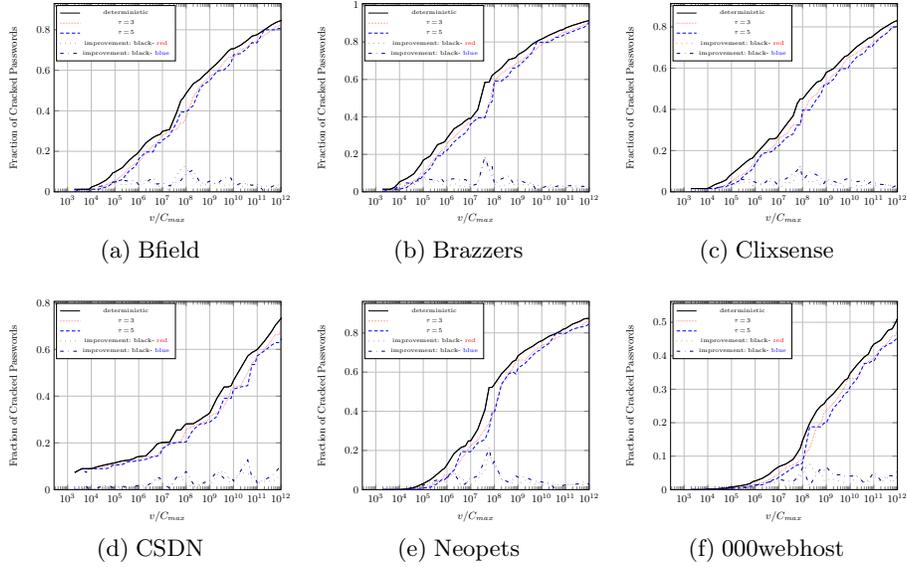
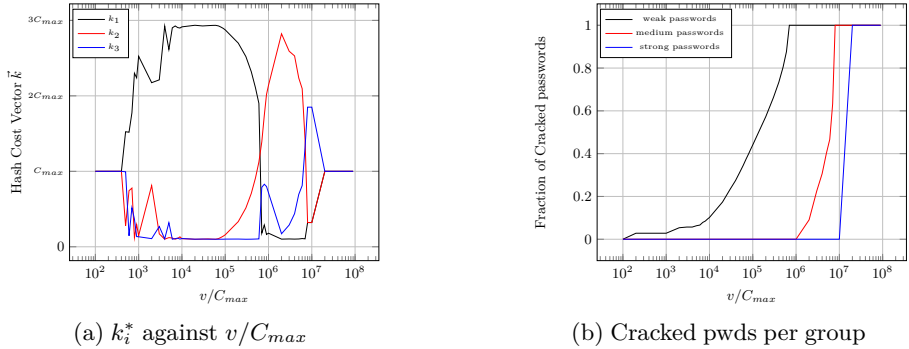
\begin{figure}[t]\centering
%
\subfloat[$k_i^*$ against $v/C_{max}$]{
\begin{tikzpicture}[scale=0.6]
\begin{semilogxaxis}[
title style={align=center},
xlabel={$v/C_{max}$},
ylabel={Hash Cost Vector $\vec{k}$},
ytick={0, 1,  2, 3},
yticklabels={0,\tiny$C_{max}$, \tiny$2C_{max}$, \tiny$3C_{max}$},
grid=major,
cycle list = {{black},  {red},  {blue}, {brown, mark=none}, {cyan, mark=none},{yellow, mark=none}, {purple, mark=none}, {red,mark=none}},
legend style = {font=\tiny, at={(.01,.98)}, anchor=north west},
legend entries = {$k_1$, $k_2$, $k_3$}
]
\addplot table[x = v, y = c1] {./sr/costs/rockyou.dat};
\addplot table[x = v, y = c2] {./sr/costs/rockyou.dat};
\addplot table[x = v, y = c3] {./sr/costs/rockyou.dat};
\end{semilogxaxis}
\end{tikzpicture}
\label{fig:r1}
}
\hfill
\subfloat[Cracked pwds per group]{
\begin{tikzpicture}[scale=0.6]
\begin{semilogxaxis}[
title style={align=center},
xlabel={$v/C_{max}$},
ylabel={Fraction of Cracked passwords},
grid=major,
cycle list = {{black},  {red},  {blue}, {brown, mark=none}, {cyan, mark=none},{yellow, mark=none}, {purple, mark=none}, {red,mark=none}},
legend style = {font=\tiny, at={(.01,.98)}, anchor=north west},
legend entries = {weak passwords, medium passwords, strong passwords}
]
\addplot table[x = v, y = c2] {./sr/pergroup/rockyoupergroup.dat};
\addplot table[x = v, y = c3] {./sr/pergroup/rockyoupergroup.dat};
\addplot table[x = v, y = c4] {./sr/pergroup/rockyoupergroup.dat};
\end{semilogxaxis}
\end{tikzpicture}
\label{fig:r2}
}
\caption{Hash Costs and Cracked Fraction per Group for RockYou (Empirical Distribution)} 
\vspace{-0.5cm}
\label{fig:rockyou}
\end{figure}

\subsection{Experiment Results}
\revision{Figure \ref{fig:empirical} evalutes the performance of DAHash on the empirical distributions empirical datasets. To generate each point on the plot we first fix $v/C_{max} \in \{ i \times 10^{2+j}: 1 \leq i \leq 9, 0 \leq j \leq 5\}$, use $\mathsf{OptHashCostVec}()$ to tune our DAHash parameters $\vec{k}^*$ and then compute the corresponding success rate for the attacker. The experiment is repeated for the empirical distributions derived from our $9$ different datasets. }
 In each experiment we group password equivalence sets into $\tau$ groups ($\tau \in \{1,3,5\}$) $G_1,\ldots,G_\tau$ of (approximately) equal probability mass. In addition, we set $k_{min} = 0.1 C_{max}$ and iteration of BITEOPT to be 10000. \revision{The yellow (resp. red) regions correspond to unconfident zones where we expect that the our results for empirical distribution might differ from reality by $1\%$ (resp. $10\%$). }

\revision{ Figure \ref{fig:monte} evaluates the performance of DAHash for for Monte Carlo distributions we extract using the Password Guessing Service. For each dataset we extract two distributions $D_{train}$ and $D_{eval}$. For each $v/C_{max} \in \{ j \times 10^i : ~ 3 \leq i \leq 11, j \in \{2,4,6,8\}\}$ we obtain the corresponding optimal hash cost $\vec{k}^*$ using $\mathsf{OptHashCostVec}()$ with the distribution $D_{train}$ as input. Then we compute success rate of attacker on $D_{eval}$ with the same cost vector $\vec{k}^*$. We repeated this for 6 plaintext datasets: Bfield, Brazzers, Clixsense, CSDN, Neopets and 000webhost for which we obtained guessing numbers from the Password Guessing Service.}

\revision{Figure \ref{fig:empirical} and  Figures \ref{fig:monte} plot $P_{ADV}$ vs $v/C_{max}$ for each different dataset under empirical distribution and Monte Carlo distribution. Each sub-figure contains three separate lines corresponding to $\tau\in \{1,3,5\}$ respectively}. We first remark that $\tau=1$ corresponds to the status quo when all passwords are assigned the same cost parameter i.e., $\mathsf{getHardness}(pw_u) = C_{max}$. When $\tau=3$ we can interpret our mechanism as classifying all passwords into three groups (e.g., weak, medium and strong) based on their strength.  The fine grained case $\tau=5$ has more strength levels into which passwords can be placed. 

\revision{{\bf \noindent DAHash Advantage:} For empirical distributions the improvement peaks in the uncertain region of the plot. Ignoring the uncertain region the improvement is still as large as 15\%. For Monte Carlo distributions we find a 20\% improvement e.g., $20\%$ of user passwords could be saved with the DAHash mechanism.}

\revision{Figure \ref{fig:r1} explores how the hash cost vector $\vec{k}$ is allocated between weak/medium/strong passwords as $v/C_{max}$ varies (using the RockYou empirical distribution with $\tau=3$). Similarly, Figure \ref{fig:r2} plots the fraction of weak/medium/strong passwords being cracked as adversary value increases. We discuss these each of these figures in more detail below.}

\vspace{-0.3cm}
\subsubsection{How Many Groups ($\tau$)?}
We explore the impact of $\tau$ on the percentage of passwords that a rational adversary will crack. Since the untargeted adversary attacks all user accounts in the very same way, the percentage of passwords the adversary will crack is the probability that the adversary succeeds in cracking a random user’s account, namely, $P_{ADV}^*$. Intuitively, a partition resulting in more groups can grant a better protection for passwords, since by doing so the authentication server can deal  with passwords with more precision and can better tune the fitness of protection level to password strength. We observe in Figure \ref{fig:empirical} and  Figures \ref{fig:monte} for most of time the success rate reduction when $\tau = 5$ is larger compared to $\tau = 3$. However, the marginal benefit plummets, changing $\tau$ from 3 to 5 does not bring much performance improvement. A positive interpretation of this observation is that we can glean most of the benefits of our  differentiated hash cost mechanism without making the $\mathsf{getHardness}()$ procedure too complicated e.g., we only need to partition passwords into three groups weak, medium and strong. 

Our hashing mechanism does not overprotect passwords that are too weak to withstand offline attack when adversary value is sufficiently high, nor passwords that are strong enough so that a rational offline attacker loses interest in cracking. The effort previously  spent in protecting passwords that are too weak/strong  can be reallocated into protecting ``savable'' passwords at some $v/C_{max}$. Thus, our DAHash algorithm beats traditional hashing algorithm without increasing the server's expected workload i.e., the cost parameters $\vec{k}$ are tuned such that expected workload is always $C_{max}$ whether $\tau=1$ (no differentiated costs), $\tau=3$ (differentiated costs) or $\tau=5$ (finer grained differentiated costs). We find that the defender can reduce the percentage of cracked passwords $P_{ADV}^*$ without increasing the workload $C_{max}$. 


\vspace{-0.3cm}
\subsubsection{Understanding the Optimal Allocation $\vec{k}^*$}
We next discuss how our mechanism re-allocates the cost parameters across $\tau=3$ different groups as $v/C_{max}$ increases --- see Figures \ref{fig:r1}. At the very beginning $v/C_{max}$ is small enough that a rational password gives up without cracking any password even if the authentication server assigns equal hash costs to different groups of password, e.g., $k_1=k_2=k_3=C_{max}$. 

As the adversary value increases the Algorithm $\mathsf{OptHashCostVec}()$ starts to reallocate $\vec{k}$ so that most of the authentication server's effort is used to protect the weakest passwords in group $G_1$ while minimal key-stretching effort is used to protect the stronger passwords in groups $G_2$ and $G_3$
 In particular, we have $k_1 \approx 3 C_{max}$ for much of the interval $v/C_{max} \in [4*10^3, 10^5]$ while $k_2,k_3$ are pretty small in this interval e.g., $k_2, k_3 \approx 0.1 \times C_{max}$.  However, as the ratio $v/C_{max}$ continues to increase from $10^6$ to $10^7$ Algorithm $\mathsf{OptHashCostVec}()$  once again begins to reallocate $\vec{k}$ to place most of the weight on $k_2$ as it is now necessary to protect passwords in group $G_2$. Over the same interval the value of $k_1$ decreases sharply as it is no longer possible to protect all of the weakest passwords group $G_1$. 

As $v/C_{max}$ continues to increase Algorithm $\mathsf{OptHashCostVec}()$  once again reallocates $\vec{k}$ to place most of the weight on $k_3$ as it is now necessary to protect the strongest passwords in group $G_3$ (and no longer possible to protect all of the medium strength passwords in group $G_2$). Finally, $v/C_{max}$ gets  too large it is no longer possible to protect passwords in any group so Algorithm $\mathsf{OptHashCostVec}()$ reverse back to equal hash costs , i.e., $k_1=k_2=k_3=C_{max}$. 

\ignore{
Notice that for Yahoo dataset (Figure \ref{fig:3a}) the zenith of $k_1$ is above $3C_{max}$. We remark that this is not a violation of server cost constraint $\sum_{i=1}^3\left(\sum_{pw_j\in G_i} \Pr[pw_j]\right) \cdot k_i\leq C_{max}$, since the password probability mass in $G_1$ is strictly less than $1/3$ (even though our group partition approach is trying to bring it to 1/3 as closely as possible). Similarly, the zenith of $k_3$ is below $3C_{max}$ that is because  passwords in $G_3$ are quantitatively dominant in dataset.}

Figures \ref{fig:r1} and \ref{fig:r2} tell a complementary story. Weak passwords are cracked first as $v/C_{max}$ increases, then follows the passwords with medium strength and the strong passwords stand until $v/C_{max}$ finally becomes sufficiently high. For example, in Figure \ref{fig:r2} we see that initially the mechanism is able to protect all passwords, weak, medium and strong. However, as $v/C_{max}$ increases from $10^5$ to $10^6$ it is no longer possible to protect the weakest passwords in group $G_1$. Up until $v/C_{max}=10^6$ the mechanism is able to protect all medium strength passwords in group $G_2$, but as the $v/C_{max}$ crosses the $10^7$ threshold it is not feasible to protect passwords in group $G_2$. The strongest passwords in group $G_3$ are completely projected until $v/C_{max}$ reaches  $2\times 10^7$ at which point it is no longer possible to protect any passwords because the adversary value is too high. 

 Viewing together with Figure \ref{fig:r1},  we observe that it is only when weak passwords are about to be cracked completely (when $v/C_{max}$ is around $7 \times 10^5$) that the authentication server begin to shift effort to protect medium passwords. The shift of protection effort continues as the adversary value increases until medium strength passwords are about to be massively cracked.  The same observation applies to medium passwords and strong password. 
While we used the plots from the RockYou dataset for discussion,  the same trends  also hold for other datasets (concrete thresholds may differ).

{\bf Robustness} 
We remark that in Figure \ref{fig:empirical} and Figure \ref{fig:monte}  the actual hash cost vector $\vec{k}$ we chose is not highly sensitive to small changes of the adversary value $v$ (only in semilog x axis fluctuation of $\vec{k}$ became obvious). Therefore, DAHash may still be useful even when it is not possible to obtain a precise estimate of $v$ \revision{or when the attacker's value $v$ varies slightly over time.}
\ignore{
  \subsubsection{Imperfect Knowledge} In real world settings the defender would not have perfect knowledge of the password distribution when partitioning passwords into $\tau$ groups $G_i$. Thus, to evaluate the performance of DAHash under more realistic settings we trained used a count min sketch data-structure to partition passwords into groups. The count min sketch is trained on the empirical password dataset e.g., RockYou or LinkedIn and provides a noisy estimate of the number of users who selected each password.  See the appendix for additional details about how we used the count min sketch data structure.}
  
{\bf Incentive Compatibility} One potential concern in assigning different hash cost parameters to different passwords is that we might inadvertently provide incentive for a user to select weaker passwords. In particular, the user might prefer a weaker password $pw_i$ to $pw_j$ ($\Pr[pw_i] > \Pr[pw_j]$) if s/he believes that the attacker will guess $pw_j$ before $pw_i$ e.g., the hash cost parameter $k(pw_j)$ is so small that makes $r_j > r_i$. We could directly encode incentive compatibility into our constraints for the feasible range of defender strategies $\mathcal{F}_{C_{max}}$  i.e., we could explicitly add a constraints that $r_j \leq r_i$ whenever $\Pr[pw_i] \leq \Pr[pw_j]$. However, Figures \ref{fig:r2} suggest that this is not necessary. Observe that the attacker does not crack any medium/high strength passwords until {\em all} weak passwords have been cracked. Similarly, the attacker does not crack any high strength passwords until {\em all} medium strength passwords have been cracked.

\section{Conclusions}

We introduce the notion of DAHash.  In our mechanism the cost parameter assigned to distinct passwords may not be the same. This allows the defender to focus key-stretching effort primarily on passwords where the effort will influence the decisions of a rational attacker who will quit attacking as soon as expected costs exceed expected rewards. We present  Stackelberg game model to capture the essentials of the interaction between the legitimate authentication server (leader) and an untargeted offline attacker (follower). In the game the defender (leader) commits to the hash cost parameters $\vec{k}$ for different passwords and the attacker responds in a utility optimizing manner. We presented a highly efficient algorithm to provably compute the attacker's best response given a  password distribution. Using this algorithm as a subroutine we use an evolutionary algorithm to find a good strategy $\vec{k}$ for the defender. Finally, we analyzed the performance of our differentiated cost password hashing algorithm using empirical password datasets . Our experiments indicate that DAHash can dramatically reduce the fraction of passwords that would be cracked in an untargeted offline attack in comparison with the traditional approach e.g., by up to  \revision{$15\%$ under empirical distributions and $20\%$ under Monte Carlo distributions}. This gain comes without increasing the expected workload of the authentication server. Our mechanism is fully compatible with modern memory hard password hashing algorithms such as SCRYPT~\cite{Per09}, Argon2id~\cite{Argon2} and DRSample~\cite{CCS:AlwBloHar17}.


\section*{Acknowledgment}
The work was supported by the National Science Foundation under grants CNS \#1704587, CNS \#1755708 and CNS \#1931443. The authors wish to thank Matteo Dell`Amico (shepherd) and other anonymous reviewers for constructive feedback which helped improve the paper.
\bibliographystyle{splncs04}
\bibliography{cryptobib/abbrev2,cryptobib/crypto,bounded-parallel-mhf,jit,extra}

\begin{thebibliography}{10}
\providecommand{\url}[1]{\texttt{#1}}
\providecommand{\urlprefix}{URL }
\providecommand{\doi}[1]{https://doi.org/#1}

\bibitem{biteopt}
Biteopt algorithm. \url{https://github.com/avaneev/biteopt}

\bibitem{CCS:Allodi17}
Allodi, L.: Economic factors of vulnerability trade and exploitation. In:
  Thuraisingham, B.M., Evans, D., Malkin, T., Xu, D. (eds.) ACM CCS 2017. pp.
  1483--1499. {ACM} Press, Dallas, TX, USA (Oct~31~--~Nov~2, 2017).
  \doi{10.1145/3133956.3133960}

\bibitem{CCS:AlwBloHar17}
Alwen, J., Blocki, J., Harsha, B.: Practical graphs for optimal side-channel
  resistant memory-hard functions. In: Thuraisingham, B.M., Evans, D., Malkin,
  T., Xu, D. (eds.) ACM CCS 2017. pp. 1001--1017. {ACM} Press, Dallas, TX, USA
  (Oct~31~--~Nov~2, 2017). \doi{10.1145/3133956.3134031}

\bibitem{Argon2}
Biryukov, A., Dinu, D., Khovratovich, D.: Argon2: new generation of memory-hard
  functions for password hashing and other applications. In: Security and
  Privacy (EuroS\&P), 2016 IEEE European Symposium on. pp. 292--302. IEEE
  (2016)

\bibitem{BlockiD16}
Blocki, J., Datta, A.: {CASH:} {A} cost asymmetric secure hash algorithm for
  optimal password protection. In: {IEEE} 29th Computer Security Foundations
  Symposium. pp. 371--386 (2016)

\bibitem{NDSS:BloDatBon16}
Blocki, J., Datta, A., Bonneau, J.: Differentially private password frequency
  lists. In: NDSS~2016. The Internet Society, San Diego, CA, USA (Feb~21--24,
  2016)

\bibitem{SP:BloHarZho18}
Blocki, J., Harsha, B., Zhou, S.: On the economics of offline password
  cracking. In: 2018 {IEEE} Symposium on Security and Privacy. pp. 853--871.
  {IEEE} Computer Society Press, San Francisco, CA, USA (May~21--23, 2018).
  \doi{10.1109/SP.2018.00009}

\bibitem{SP:Bonneau12}
Bonneau, J.: The science of guessing: Analyzing an anonymized corpus of 70
  million passwords. In: 2012 {IEEE} Symposium on Security and Privacy. pp.
  538--552. {IEEE} Computer Society Press, San Francisco, CA, USA (May~21--23,
  2012). \doi{10.1109/SP.2012.49}

\bibitem{USENIX:Boyen07}
Boyen, X.: Halting password puzzles: Hard-to-break encryption from
  human-memorable keys. In: Provos, N. (ed.) USENIX Security 2007. {USENIX}
  Association, Boston, MA, USA (Aug~6--10, 2007)

\bibitem{Castelluccia2013}
Castelluccia, C., Chaabane, A., D{\"u}rmuth, M., Perito, D.: When privacy meets
  security: Leveraging personal information for password cracking. arXiv
  preprint arXiv:1304.6584  (2013)

\bibitem{NDSS:CasDurPer12}
Castelluccia, C., D{\"u}rmuth, M., Perito, D.: Adaptive password-strength
  meters from {Markov} models. In: NDSS~2012. The Internet Society, San Diego,
  CA, USA (Feb~5--8, 2012)

\bibitem{CCS:DelFil15}
Dell'Amico, M., Filippone, M.: Monte carlo strength evaluation: Fast and
  reliable password checking. In: Ray, I., Li, N., Kruegel, C. (eds.) ACM CCS
  2015. pp. 158--169. {ACM} Press, Denver, CO, USA (Oct~12--16, 2015).
  \doi{10.1145/2810103.2813631}

\bibitem{EC:DodGuoKat17}
Dodis, Y., Guo, S., Katz, J.: Fixing cracks in the concrete: Random oracles
  with auxiliary input, revisited. In: Coron, J., Nielsen, J.B. (eds.)
  EUROCRYPT~2017, Part~II. {LNCS}, vol. 10211, pp. 473--495. Springer,
  Heidelberg, Germany, Paris, France (Apr~30~--~May~4, 2017).
  \doi{10.1007/978-3-319-56614-6_16}

\bibitem{passwordBlackMarket}
Fossi, M., Johnson, E., Turner, D., Mack, T., Blackbird, J., McKinney, D., Low,
  M.K., Adams, T., Laucht, M.P., Gough, J.: Symantec report on the underground
  economy  (November 2008), retrieved 1/8/2013.

\bibitem{CS:HMBSD20}
Harsha, B., Morton, R., Blocki, J., Springer, J., Dark, M.: Bicycle attacks
  considered harmful: Quantifying the damage of widespread password length
  leakage. Computers \& Security  \textbf{100},  102068 (2021).
  \doi{https://doi.org/10.1016/j.cose.2020.102068},
  \url{http://www.sciencedirect.com/science/article/pii/S0167404820303412}

\bibitem{goldForSilver}
Herley, C., Flor{\^e}ncio, D.: Nobody sells gold for the price of silver:
  Dishonesty, uncertainty and the underground economy. Economics of information
  security and privacy pp. 33--53 (2010)

\bibitem{kaliski2000pkcs}
Kaliski, B.: Pkcs\# 5: Password-based cryptography specification version 2.0
  (2000)

\bibitem{SP:KKMSVB12}
Kelley, P.G., Komanduri, S., Mazurek, M.L., Shay, R., Vidas, T., Bauer, L.,
  Christin, N., Cranor, L.F., Lopez, J.: Guess again (and again and again):
  Measuring password strength by simulating password-cracking algorithms. In:
  2012 {IEEE} Symposium on Security and Privacy. pp. 523--537. {IEEE} Computer
  Society Press, San Francisco, CA, USA (May~21--23, 2012).
  \doi{10.1109/SP.2012.38}

\bibitem{SP:MYLL14}
Ma, J., Yang, W., Luo, M., Li, N.: A study of probabilistic password models.
  In: 2014 {IEEE} Symposium on Security and Privacy. pp. 689--704. {IEEE}
  Computer Society Press, Berkeley, CA, USA (May~18--21, 2014).
  \doi{10.1109/SP.2014.50}

\bibitem{manber1996simple}
Manber, U.: A simple scheme to make passwords based on one-way functions much
  harder to crack. Computers \& Security  \textbf{15}(2),  171--176 (1996)

\bibitem{USENIX:MUSKBCC16}
Melicher, W., Ur, B., Segreti, S.M., Komanduri, S., Bauer, L., Christin, N.,
  Cranor, L.F.: Fast, lean, and accurate: Modeling password guessability using
  neural networks. In: Holz, T., Savage, S. (eds.) USENIX Security 2016. pp.
  175--191. {USENIX} Association, Austin, TX, USA (Aug~10--12, 2016)

\bibitem{Morris1979}
Morris, R., Thompson, K.: Password security: A case history. Communications of
  the ACM  \textbf{22}(11),  594--597 (1979),
  \url{http://dl.acm.org/citation.cfm?id=359172}

\bibitem{C:Oechslin03}
Oechslin, P.: Making a faster cryptanalytic time-memory trade-off. In: Boneh,
  D. (ed.) CRYPTO~2003. {LNCS}, vol.~2729, pp. 617--630. Springer, Heidelberg,
  Germany, Santa Barbara, CA, USA (Aug~17--21, 2003).
  \doi{10.1007/978-3-540-45146-4_36}

\bibitem{Per09}
Percival, C.: Stronger key derivation via sequential memory-hard functions. In:
  BSDCan 2009 (2009)

\bibitem{provos1999bcrypt}
Provos, N., Mazieres, D.: Bcrypt algorithm. USENIX (1999)

\bibitem{rios2013derivative}
Rios, L.M., Sahinidis, N.V.: Derivative-free optimization: a review of
  algorithms and comparison of software implementations. Journal of Global
  Optimization  \textbf{56}(3),  1247--1293 (2013)

\bibitem{stockley2016}
Stockley, M.: What your hacked account is worth on the dark web (Aug 2016),
  \url{https://nakedsecurity.sophos.com/2016/08/09/what-your-hacked-account-is-worth-on-the-dark-web/}

\bibitem{USENIX:USBCCKKMMS15}
Ur, B., Segreti, S.M., Bauer, L., Christin, N., Cranor, L.F., Komanduri, S.,
  Kurilova, D., Mazurek, M.L., Melicher, W., Shay, R.: Measuring real-world
  accuracies and biases in modeling password guessability. In: Jung, J., Holz,
  T. (eds.) USENIX Security 2015. pp. 463--481. {USENIX} Association,
  Washington, DC, USA (Aug~12--14, 2015)

\bibitem{FC:VBCKM16}
Vasek, M., Bonneau, J., Castellucci, R., Keith, C., Moore, T.: The bitcoin
  brain drain: Examining the use and abuse of bitcoin brain wallets. In:
  Grossklags, J., Preneel, B. (eds.) FC 2016. {LNCS}, vol.~9603, pp. 609--618.
  Springer, Heidelberg, Germany, Christ Church, Barbados (Feb~22--26, 2016)

\bibitem{NDSS:VerColTho14}
Veras, R., Collins, C., Thorpe, J.: On semantic patterns of passwords and their
  security impact. In: NDSS~2014. The Internet Society, San Diego, CA, USA
  (Feb~23--26, 2014)

\bibitem{von2010market}
Von~Stackelberg, H.: Market structure and equilibrium. Springer Science \&
  Business Media (2010)

\bibitem{SP:WAMG09}
Weir, M., Aggarwal, S., de~Medeiros, B., Glodek, B.: Password cracking using
  probabilistic context-free grammars. In: 2009 {IEEE} Symposium on Security
  and Privacy. pp. 391--405. {IEEE} Computer Society Press, Oakland, CA, USA
  (May~17--20, 2009). \doi{10.1109/SP.2009.8}

\bibitem{EPRINT:Wetzels16}
Wetzels, J.: Open sesame: The password hashing competition and {Argon2}.
  Cryptology ePrint Archive, Report 2016/104 (2016),
  \url{http://eprint.iacr.org/2016/104}

\bibitem{JC:Wiener04}
Wiener, M.J.: The full cost of cryptanalytic attacks. Journal of Cryptology
  \textbf{17}(2),  105--124 (Mar 2004). \doi{10.1007/s00145-003-0213-5}

\end{thebibliography}
\appendix

\section{Algorithms}\label{app:algorithm}
\begin{algorithm}[h]
\begin{algorithmic}[1]
\Require{$u$, $pw_u$, $L$}
\State $s_u \overset{\$}{\leftarrow} \{0,1\}^L$;
\State $k\leftarrow \mathsf{GetHardness}(pw_u)$;
\State $h\leftarrow H(pw_u,s_u;~k)$;
\State $\mathsf{StoreRecord}$ $(u,s_u,h)$
\end{algorithmic}
\caption{Account creation}  \label{alg:createaccount} 
\end{algorithm}

\begin{algorithm}[h]
\begin{algorithmic}[1]
\Require{$u$, $pw_u'$}
\State $(u,s_u,h) \leftarrow \mathsf{FindRecord}(u)$;
\State $k' \leftarrow \mathsf{GetHardness}(pw_u')$;
\State $h'\leftarrow H(pw_u,s_u;~k')$;
\State \textbf{Return} $h == h'$
\end{algorithmic}
\caption{Password authentication} \label{alg:authenticate}
\end{algorithm}
\vspace{-0.3cm}
\ignore{
\begin{algorithm}[h]
\begin{algorithmic}[1]
\State \textbf{Preprocessing:}
\State Partition all passwords into $\tau$ groups $G_i,i\in\{1,\cdots,\tau\}$;
\State Associate $k_i$ with $G_i$;
\Require{$pw_u$}
\Ensure{$k(pw_i)$}
\For{$i=1$ to $\tau$}
\If{$pw_u\in G_i$}
\State $k(pw_u)=k_i$;  
\EndIf
\EndFor
\State \textbf{return} $k(pw_i)$;
\end{algorithmic}
 \caption{$\mathsf{GetHardness}(pw_u)$} \label{alg:getHardness}
 \label{alg:hardness}
\end{algorithm}
}
\vspace{-0.3cm}
\begin{algorithm}[h]
\caption{The adversary’s best response $\mathsf{BestRes}(v, \vec{k}, D), $}
\begin{algorithmic}[1]
\Require{$\vec{k}$, $v$, $D$}
\Ensure{$(\pi^*, B^*)$}
\State sort $\{\frac{p_i}{k_i}\}$ and reindex such that  $\frac{p_1}{k_1}\geq\cdots\geq\frac{p_{n’}}{k_{n’}}$ to get $\pi^*$;
\State $B^* = \arg \max U_{ADV}\left(v,\vec{k}, (\pi^*,B)\right)$
\State \textbf{return} $(\pi^*,B^*)$;
\end{algorithmic}
\label{alg:response}
\end{algorithm}

\section{Missing Proofs}\label{app:proof}

\subsection*{Proof of Theorem\ref{thm:noinversions}}
\revision{
\begin{remindertheorem}{Theorem \ref{thm:noinversions}}
\thmnoinversions
\end{remindertheorem}}
\begin{proofof}{Theorem\ref{thm:noinversions}}
\revision{Fixing $B,v,\vec{k}$ we let $\pi$ be the optimal ordering of passwords. If there are multiple optimal orderings we take the ordering $\pi$ with the fewest number of inversions. Recall that an inversion is a pair $b< a$ such that $r_{\pi(a)}>r_{\pi(b)}$ i.e., $pw_{\pi(b)}$ is scheduled to be checked before $pw_{\pi(a)}$ but password $pw_{\pi(a)}$ has a higher ``bang-for-buck'' ratio. We say that we have a consecutive inversion if $a=b+1$. Suppose for contradiction that $\pi$ has an inversion  }
\begin{itemize}
\item \revision{If $\pi$ has an inversion then $\pi$ also has a consecutive inversion. Let $(a,b)$ be the closest inversion i.e., minimizing $|a-b|$. The claim is that $(a,b)$ is a consecutive inversion. If not there is some $c$ such that $b < c < a$. Now either $r_{\pi(c)} < r_{\pi(a)}$ (in which case the pair $(c,a)$ form a closer inversion) or $r_{\pi(c)} \geq r_{\pi(a)} > r_{\pi(b)}$ (in which case the pair $(b,c)$ forms a closer inversion). In either case we contradict our assumption.  } 

\item \revision{Let $b$, $b+1$ be a consecutive inversion.} We now define $\pi'$ to be the same ordering as $\pi$ except that the order of \revision{$b$ and $b+1$} is flipped \revision{i.e., $\pi'(b) = \pi(b+1)$ and $\pi'(b+1)=\pi(b)$ so that we now check password $pw_{\pi(b+1)}$ before password $pw_{\pi(b)}$}. Note that $\pi'$ has one fewer inversion than $\pi$. 

\item We will prove that $$U_{ADV}\left(v,\vec{k},(\pi',B)\right) \geq U_{ADV}\left(v,\vec{k},(\pi,B)\right)$$ contradicting the choice of $\pi$ as the optimal ordering with the fewest number of inversions. By definition \eqref{eq:utility} we have 
\begin{equation*}\small
\begin{aligned}
&U_{ADV}\left(v,\vec{k},(\pi,B)\right) =v\cdot \lambda(\pi,B)-\sum^ B_{i=1} k(pw_{\pi(i)})\cdot \left(1-\lambda(\pi,i-1)\right),
\end{aligned}
\end{equation*}
and 
\begin{equation*}\small
\begin{aligned}
&U_{ADV}\left(v,\vec{k},(\pi',B)\right) =v\cdot \lambda(\pi’,B)-\sum^ B_{i=1} k(pw_{\pi’(i)})\cdot \left(1-\lambda(\pi’,i-1)\right).
\end{aligned}
\end{equation*}
Note that $\pi$ and $\pi’$ only differ at \revision{guesses $b$ and $b+1$ and coincide at the rest of passwords.} Thus, we have  
$\lambda(\pi,i)=\lambda(\pi',i)$ when $0\leq i\leq b-1$ or when $i \geq b+1$. 
\revision{For convenience, set} $\lambda=\lambda(\pi,b-1)$. 

\revision{Assuming that $b+1 \leq B$ and taking difference of above two equations, }
\begin{equation}\label{eq:diff}
\begin{aligned}
&U_{ADV}\left(v,\vec{k},(\pi,B)\right) -U_{ADV}\left(v,\vec{k},(\pi',B)\right)\\
&=k(pw_{\pi(b)})\lambda+k(pw_{\pi(b+1)})(\lambda+p_{\pi(b)})\\
&-k(pw_{\pi(b+1)})\lambda+k(pw_{\pi(b)})(\lambda+p_{\pi(b+1)})\\
&=p_{\pi(b)}\cdot k(pw_{\pi(b+1)})- p_{\pi(b+1)}\cdot k(pw_{\pi(b)}) \leq 0.
\end{aligned}
\end{equation}

\revision{The last inequality holds since $0> (r_{\pi(b)} - r_{\pi(b+1)}) = \frac{p_{\pi(b)}}{k(pw_{\pi(b)})}  - \frac{p_{\pi(b+1)}}{k(pw_{\pi(b+1)})}$ (we multiply by both sides of the inequality by $ \left(k(pw_{\pi(b+1)}) k(pw_{\pi(b)})\right)$ to obtain the result).  From equation \eqref{eq:diff} we see that the new swapped strategy $\pi’$ has a utility at least as large as $\pi$. Contradiction! }

\revision{ If $b> B$ then swapping has no impact on utility as neither password $pw_{\pi(b)}$ or $pw_{\pi(b+1)}$ will be checked. 

Finally if $B=b$ then checking last password in $\pi$ provides non-negative utility, i.e.,
\begin{equation}
v\cdot p_{\pi(B)} - k(pw_{\pi(B)})(1-\lambda(\pi,B-1)) \geq 0,
\end{equation}
whereas continue to  check $pw(B+1)$ after executing strategy $(\pi,B)$ would reduce utility, i.e.,
\begin{equation}
v\cdot p_{\pi(B+1)} - k(pw_{\pi(B+1)})(1-\lambda(\pi,B)) < 0.
\end{equation}
From the above two equations, we have
\begin{equation}
r_{\pi(B)}=\frac{p_{\pi(B)}}{k(pw_{\pi(B)})} \geq \frac{1-\lambda(\pi,B-1)}{v} > \frac{1-\lambda(\pi,B)}{v} >  \frac{p_{\pi(B+1)}}{k(pw_{\pi(B)+1})}  = r_{\pi(B+1)}.
\end{equation}
Again, we have contradiction. Therefore, an optimal checking sequence does not contain inversions.
}
\end{itemize} 
\end{proofof}

\subsection*{Proof of Theorem \ref{corollary}}
\revision{
\begin{remindertheorem}{Theorem \ref{corollary}}
\maincorollary
\end{remindertheorem}
}

\begin{proofof}{Theorem \ref{corollary}}
The proof of Theorem \ref{corollary} follows from the following lemma which states that whenever $pwd_i$ and $pwd_j$ are in the same equivalence set the optimal attack strategy will either check both of these passwords or neither.
\begin{lemma}
\thmcompact
\end{lemma}

\begin{proof}
Suppose for contradiction that the optimal strategy checks  $pwd_i$ but not $pwd_j$. Then WLOG we can assume that $\mathsf{Inv}_{\pi^*}(i)= B^*$ is the last password to be checked and that $\mathsf{Inv}_{\pi^*}(j) = B^*+1$ is the next password to be checked (otherwise, we can swap $pwd_j$ with the password in the equivalence set that will be checked next). Since $pw_i$ and $pwd_j$ are in the same equivalence set, we have $\Pr[pw_i]=\Pr[pw_j]$ and $k(pw_i)=k(pw_j)$. The marginal utility of checking $pwd_i$ is
$$\Delta_i=v\Pr[pw_i]-k(pw_i)(1-\lambda(\pi^*,B^*)).$$
Because checking $pwd_i$ is part of the optimal strategy, it must be the case $\Delta_i\geq 0$. Otherwise, we would immediately derive a contradiction since the strategy $(\pi^*,B^*-1)$ would have greater utility than $(\pi^*,B^*)$. Now the marginal utility $\Delta_j = U_{ADV}\left(v,\vec{k},(\pi^*,B^*+1)\right) -U_{ADV}\left(v,\vec{k},(\pi^*,B^*)\right)$  of checking $pw_j$ as well is
$$\Delta_j=v\Pr[pw_j]-k(pw_j)(1-\lambda(\pi,B^*)-\Pr[pw_j])>\Delta_i \geq 0 \ . $$
Since $\Delta_j>0$ we have $U_{ADV}\left(v,\vec{k},(\pi^*,B^*+1)\right) > U_{ADV}\left(v,\vec{k},(\pi^*,B^*)\right) $ contradicting the optimality of $(\pi^*,B^*)$. \hfill $\square$
\end{proof}

From Theorem \ref{thm:noinversions} it follows that we will check the equivalence sets in the order of bang-for-buck ratios. Thus, $B^*$ must lie in the set $\{0,|es_1|,|es_1|+|es_2|,\ldots, \sum_{i=1}^{n^{\prime}} |es_i|\}$. 
\end{proofof}

\section{FAQ}

\subsection*{Could this mechanism harm user's who pick weak passwords?} \revision{
We understand the concern that our mechanism might provide less protection for weak passwords since we using a uniform hash cost for all passwords. If our estimation of the value $v$ of a cracked password is way too high then it is indeed possible that the DAHash parameters would be misconfigured in a way that harms users with weak passwords. However, even in this case we ensure that every password recieves a minimum level of acceptable protection by setting a minimum hash cost parameter $k_{min}$ for any password. We note that if our estimation of $v$ is accurate and it is feasible to deter an attacker from cracking weaker passwords then DAHash will actually tend to provide stronger protection for these passwords. On the other hand if the password is sufficiently weak that we cannot deter an attacker then these weak passwords will always be cracked no matter what actions we take. Thus, DAHash will reallocate effort to focus on protecting stronger passwords. }

\end{document}